\documentclass[11pt]{article}

\usepackage[margin=1in]{geometry}
\usepackage{fullpage}
\usepackage{amsmath,amsthm,amssymb,bm}
\usepackage{mleftright}
\usepackage{algorithm}
\usepackage{algpseudocode}
\usepackage{textcomp}
\usepackage{array}
\usepackage{hyperref,url,xcolor}

\newcommand{\oA}{\bar{A}}
\newcommand{\oB}{\overline{B}}
\newcommand{\getMinor}{\operatorname{getMinor}}
\newcommand{\coker}{\operatorname{coker}}
\newcommand{\rMDS}{\operatorname{rMDS}}
\newcommand{\rLDMDS}{\operatorname{rLD-MDS}}
\newcommand{\rMR}{\operatorname{rMR}}
\newcommand{\MR}{\operatorname{MR}}

\newcommand{\MDS}{\operatorname{MDS}}

\newcommand{\LDMDS}{\operatorname{LD-MDS}}
\newcommand{\wt}{\operatorname{wt}}
\newcommand{\rank}{\operatorname{rank}}
\newcommand{\supp}{\operatorname{supp}}
\newcommand{\eps}{\epsilon}

\newcommand{\F}{\mathbb F}

\newcommand{\Proj}{\mathbb P}

\renewcommand{\P}{\mathbb P}

\newcommand{\cG}{\mathcal G}
\newcommand{\cH}{\mathcal H}

\newcommand{\balpha}{{\bm \alpha}}

\newcommand{\mult}{\mathrm{mult}}

\newcommand{\Eval}{\operatorname{Eval}}

\newcommand{\divi}{\operatorname{div}}
\newcommand{\AG}{\operatorname{AG}}

\newtheorem{theorem}{Theorem}[section]

\newtheorem{lemma}[theorem]{Lemma}
\newtheorem{corollary}[theorem]{Corollary}
\newtheorem{proposition}[theorem]{Proposition}

\newtheorem{claim}[theorem]{Claim}
\newtheorem{definition}[theorem]{Definition}
\newtheorem{remark}{Remark}[section]

\theoremstyle{definition}
\newtheorem{framework}{Framework}[section]

\hypersetup{
	colorlinks=true,
	linkcolor=blue,%
	citecolor=blue,
	linktoc=page
}

\usepackage{thmtools}
\usepackage{thm-restate}

\title{AG Codes Achieve List-decoding Capacity\\over Constant-sized Fields\thanks{An abridged version of this manuscript appeared in the 2024 \emph{ACM Symposium on Theory of Computing (STOC 2024)}. This full version includes proofs of all claims and an expanded discussion in Section 4.}}

\author{Joshua Brakensiek\thanks{University of California, Berkeley. Email: \texttt{josh.brakensiek@berkeley.edu}. Research supported in part by a Microsoft Research PhD fellowship while a student at Stanford University.} \and
Manik Dhar\thanks{Department of Mathematics, Massachusetts Institute of Technology. Email: \texttt{dmanik@mit.edu}. Part of this work was done while this author was a graduate student at the Department of Computer Science, Princeton University where his research was supported by NSF grant DMS-1953807.} \and
Sivakanth Gopi\thanks{Microsoft Research. Email: \texttt{sigopi@microsoft.com}} \and Zihan Zhang\thanks{Department of Computer Science and Engineering, The Ohio State University. Email: \texttt{zhang.13691@osu.edu}}}
\date{}

\begin{document}
\maketitle

\begin{abstract}
The recently-emerging field of higher order MDS codes has sought to unify a number of concepts in coding theory. Such areas captured by higher order MDS codes include maximally recoverable (MR) tensor codes, codes with optimal list-decoding guarantees, and codes with constrained generator matrices (as in the GM-MDS theorem).

By proving these equivalences, Brakensiek-Gopi-Makam (\cite{bgm2022}) showed the existence of optimally list-decodable Reed-Solomon codes over exponential sized fields.
Building on this, recent breakthroughs by Guo-Zhang (\cite{guo2023randomly}) and Alrabiah-Guruswami-Li (\cite{alrabiah2023randomly}) have shown that randomly punctured Reed-Solomon codes achieve list-decoding capacity (which is a relaxation of optimal list-decodability) over linear size fields. We extend these works by developing a formal theory of \emph{relaxed higher order MDS codes}. In particular, we show that there are two inequivalent relaxations which we call \emph{lower} and \emph{upper} relaxations. The lower relaxation is equivalent to relaxed optimal list-decodable codes and the upper relaxation is equivalent to relaxed MR tensor codes with a single parity check per column.

We then generalize the techniques of Guo-Zhang and Alrabiah-Guruswami-Li to show that both these relaxations can be constructed by randomly puncturing suitable algebraic-geometric codes over \emph{constant size} fields. For this, we crucially use the generalized GM-MDS theorem for polynomial codes recently proved by Brakensiek-Dhar-Gopi (\cite{bdg2023a}). We obtain the following corollaries from our main result:
\begin{itemize}
	 \item Randomly punctured algebraic-geometric codes of rate $R$ are list-decodable up to radius $\frac{L}{L+1}(1-R-\eps)$ with list size $L$ over fields of size $\exp(O(L/\eps))$. In particular, they achieve list-decoding capacity with list size $O(1/\eps)$ and field size $\exp(O(1/\eps^2))$. Prior to this work, AG codes were not even known to achieve list-decoding capacity.
	 \item By randomly puncturing algebraic-geometric codes, we can construct relaxed MR tensor codes with a single parity check per column over \emph{constant-sized} fields, whereas (non-relaxed) MR tensor codes require exponential field size.
\end{itemize}
\end{abstract}

\newpage
\tableofcontents
\newpage
\section{Introduction}
MDS (maximum distance separable) codes meet the Singleton bound \cite{singleton1964maximum}, which is the optimal rate-distance tradeoff for codes over large alphabet. It states that an $(n,k)$-code has distance $d\le n-k+1$. Reed-Solomon codes \cite{reed1960polynomial} are a simple and explicit construction of linear MDS codes over fields of linear size (i.e., $q=O(n)$). Due to their optimality, MDS codes are extensively used in practice such as in communication, data storage, cryptography etc.. An $(n,k)$-code is MDS iff any $k$ columns of its generator matrix are linearly independent.\footnote{In this paper, we will only focus on linear codes. Unless explicitly mentioned, all codes are linear.} Higher order MDS codes were recently introduced by Brakensiek-Gopi-Makam (\cite{bgm2021mds}) as a natural generalization of MDS codes.
\begin{definition}\label{def:mds}
    An $(n,k)$-code with generator matrix $V_{k\times n}$ is order-$\ell$ higher order MDS, denoted by $\MDS(\ell)$, if for every $\ell$ subsets $A_1,A_2,\dots,A_\ell \subset [n]$, we have $\dim(V_{A_1}\cap V_{A_2} \cap \dots \cap V_{A_\ell}) = \dim(W_{A_1}\cap W_{A_2} \cap \dots \cap W_{A_\ell})$ where $W_{k \times n}$ is a generic\footnote{One suitable definition of generic is that the entries of $W$ are chosen randomly from a field of sufficiently large size, like $q^{nk}$.} $k\times n$ matrix. 
\end{definition}
Here $V_{A}$ is the span of the columns indexed by $A$. In other words, $\ell$ subspaces spanned by subsets of columns should intersect as minimally as possible. When $\ell\le 2$, $\MDS(1)$ and $\MDS(2)$ codes are equivalent to MDS codes (\cite{bgm2021mds}).
Higher order MDS codes have since been shown to be equivalent to many concepts independently studied in coding theory (\cite{bgm2022}). These include maximally recoverable (MR) tensor codes, codes with optimal list-decoding guarantees, and codes with constrained generator matrices (as in the GM-MDS theorem). We refer the reader to \cite{bgm2022} for a detailed survey of these connections. \cite{bgm2022} showed that random Reed-Solomon codes over exponentially large fields are higher order MDS codes with high probability. A major challenge is to give (explicit) constructions of higher order MDS codes over small fields. But in a recent work, \cite{bdg2023size} showed that even $\MDS(3)$ codes of constant rate require exponential field size, in sharp contrast to MDS codes (or $\MDS(2)$ codes) for which we have explicit constructions over linear size fields (for example Reed-Solomon codes). Therefore, it is natural to look for some kind of relaxation of higher order MDS codes which would allow for constructions over small fields. In this work, we introduce two ways to relax higher order MDS codes called \emph{lower} and \emph{upper} relaxation (these are defined in Section~\ref{sec:relaxed}). We then establish equivalences between these relaxations and suitable relaxations of optimal list-decodable codes and MR tensor codes. In particular, the lower relaxation is equivalent to relaxed version of optimal list-decodable codes and the upper relaxation is equivalent to the relaxation of MR tensor codes with a single parity check per column. Finally we show that one can indeed construct such (upper and lower) relaxed higher order MDS codes over small fields. We start by briefly defining some of these concepts and discuss prior work along the way.

\paragraph{Optimal list-decodable codes ($\LDMDS(\le L)$)} List-decoding is an important concept in coding theory which allows for correction beyond half the minimum distance \cite{elias1957list,wozencraft1958list}. An $(n,k)$ code $C$ is $(\rho,L)$-list decodable if any Hamming ball of radius $\rho n$ contains at most $L$ codewords. The generalized Singleton bound due to \cite{shangguan2020combinatorial,roth2021higher,goldberg2021singleton} generalizes the well-known Singleton bound to the setting of list-decoding. The generalized Singleton bound states that for every $(n,k)$ code $C$ that is $(\rho,L)$-list decodable, we have 
\begin{equation}
    \label{eq:gen-singleton}
    \rho \le \frac{L}{L+1}(1-k/n).
\end{equation}
Note that when $L=1$, we recover the Singleton bound. The same bound also holds for average-radius list-decoding which is a strengthening of list-decoding. An $(n,k)$-code $C$ is $(\rho,L)$-average-radius list-decodable if for any $y\in \F^n$ and any $L+1$ codewords $c_0,c_1,\dots,c_L\in C$, we have that $\frac{1}{L+1}\sum_{i=0}^L \wt(y-c_i) \le \rho n$. We now define the notion of optimal list-decodable codes.
\begin{definition}
    We say that $C$ is $\LDMDS(L)$ if it is $(\rho,L)$-average-radius list-decodable with radius $\rho = \frac{L}{L+1}(1-k/n)$, matching the generalized Singleton bound (\ref{eq:gen-singleton}). 
\end{definition}
A code is $\LDMDS(\le L)$ if it is $\LDMDS(\ell)$ for all list sizes $\ell \le L$. We would like to have explicit constructions of $\LDMDS(\le L)$ codes over small alphabets. \cite{bgm2022} showed that random Reed-Solomon codes over exponentially large fields are $\LDMDS(\le L)$, i.e., optimally list-decodable for all list sizes $L$, proving a conjecture of \cite{shangguan2020combinatorial}. Meanwhile, as discussed before, \cite{bdg2023size} showed that even $\LDMDS(\le 2)$ codes of constant rate, i.e., optimally list-decodable with list size 2, require exponential field size. Thus it is natural to look at relaxations of optimal list-decodability if we want constructions over small fields. A natural relaxation is to ask that a rate $R$ code is $(\rho,L)$-average-radius list-decodable for 

\begin{equation}
    \label{eq:relaxedrho}
    \rho=\frac{L}{L+1}(1-R-\eps)    
\end{equation}
 for some $\eps>0$. We call such codes \emph{relaxed $\LDMDS$ codes} (see Section~\ref{sec:relaxed} for a formal definition) and denote them by $\rLDMDS$. We show that these relaxed $\LDMDS$ codes are equivalent to the lower relaxation of higher order MDS codes. 
 
 Guo-Zhang (\cite{guo2023randomly}) were the first to observe that such a relaxation leads to much improved constructions. They showed that random Reed-Solomon codes over fields of size $O_{L,\eps}(n^2)$ are relaxed $\LDMDS$ codes with $\rho$ given by (\ref{eq:relaxedrho}).\footnote{Here random Reed-Solomon codes over $\F_q$ refers to choosing the $n$ evaluation points at random from $\F_q$. Alternatively, one can think of it as the code of length $n$ obtained by randomly puncturing Reed-Solomon code of length $q$ evaluated at all points of $\F_q$.} Alrabiah-Guruswami-Li (\cite{alrabiah2023randomly}) further improved the field size to $O_{L,\eps}(n)$ for random Reed-Solomon codes. They also showed that random linear codes achieve the same relaxed list-decoding radius in (\ref{eq:relaxedrho}) with $\exp(O(L/\eps))$ field size. In subsequent work \cite{alrabiah2023ag}, they also showed a lower bound that any (not necessarily linear) code which achieves the bound in (\ref{eq:relaxedrho}) must have field size at least $\exp(\Omega_{L,R}(1/\eps))$. For average-radius list-decoding, they obtain a better lower bound of $\exp(\Omega_{R}(1/\eps))$ independent of list size $L\ge 2$.

One can further relax the list-decodability requirement to get what is known as \emph{list-decoding capacity achieving codes}. Note that as $L\to \infty$, the list-decoding radius $\rho \to  1-R$ where $R=k/n$ is the rate of the code, i.e., any non-trivial list-decoding cannot be done beyond $\rho=1-R$. This is the called the \emph{list-decoding capacity}. We say that a code family achieves list-decoding capacity if for every $\eps>0$ and $R\in (0,1)$, there exists a code $C$ from this family of rate $R$ which is $(\rho,L)$-list decodable with $\rho = 1 - R - \eps$ and $L=O_\eps(1)$.\footnote{In some works, even getting $L=n^{O_\eps(1)}$ is considered enough.} In addition, we would also want the alphabet size of the code to be as small as possible. Random codes of rate $R$ are $(1-R-\eps,O(1/\eps))$-list-decodable (with alphabet size $2^{O(1/\epsilon)}$, see \cite{guruswami2012essential}), this was recently shown to hold for random linear codes as well (with larger alphabet size $2^{O(1/\eps^2)}$) in \cite{alrabiah2023randomly}. There is a long line of work trying to construct explicit code families achieving list-decoding capacity. After an initial breakthrough by \cite{parvaresh2005correcting}, \cite{guruswami2008explicit} constructed the first known family of codes achieving list-decoding capacity called Folded Reed-Solomon codes. Though the initial list size and alphabet size were of the form $n^{O_\eps(1)}$, later works have produced many constructions which have reduced the alphabet size and list size to $\exp(poly(1/\eps))$ (or lower) \cite{dvir2012subspace,guruswami2012folded,guruswami2013list,guruswami2013linear, kopparty2015list,guruswami2016explicit,rudra2017average,kopparty2018improved,hemenway2019local,guruswami2022optimal,guo2021efficient,srivastava2025improved,cz24}. See \cite{guo2021efficient} for an explicit construction matching these bounds. They obtain this code by starting with an AG code over $\F_{q^m}$ evaluated only on points of $\F_q$. Then they take a subcode of this by restricting the messages to a BTT evasive subspace (where BTT stands for Block Triangular Toeplitz). Similar to folded Reed-Solomon codes, their code is not linear over $\F_{q^m}$ which is the alphabet--it is only linear over the base field $\F_q$. They also give an efficient list-decoding algorithm.  \cite{guruswami2022optimal} give a randomized construction of codes achieving list-decoding capacity with alphabet size $\exp(\tilde{O}(1/\eps^2))$ and a much better list size of $O(1/\eps)$. They obtain this by folding AG codes using automorphisms of the underlying function field. The message space is also restricted to a hierarchical subspace-evasive set (of which we don't have explicit constructions, so they give a pseudorandom construction) to obtain the list size of $O(1/\eps)$. In particular, their codes are non-linear because of this restriction. But they give an efficient list-decoding algorithm for their codes.

We show that by simple random puncturing of an AG code, we can achieve relaxed $\LDMDS$ codes over constant size fields. In fact, the field size is also nearly optimal and matches the lower bound shown in \cite{alrabiah2023ag}. This result can be thought of as a partial derandomization of random linear codes which achieve the same bounds.

\begin{theorem}{(Informal - See Theorem~\ref{thm:AGpuncture} and Corollary~\ref{cor:AGListcap})}
    \label{thm:relaxedldmds_informal}
    By randomly puncturing suitable AG codes, with high probability, we get a code of rate $R$ which is $(\rho,L)$-list-decodable with $\rho = \frac{L}{L+1}(1-R-\eps)$ (as in (\ref{eq:relaxedrho})) over fields of size $\exp(O(L/\eps))$. The same result also holds for the stronger average-radius list-decoding.
\end{theorem}

We summarize the prior results and our work in Table~\ref{tab:listdecoding}.

\begin{table}[h]
\centering
\begin{tabular}{|m{0.16 \linewidth}|c|c|c|c|c|}
\hline
\textbf{Code Family} & \textbf{List Size} & \textbf{Radius} ($\rho$)  & \textbf{Field Size} & \textbf{Construction} & \textbf{Reference} \\
\hline
Random non-linear code & $O(1/\eps)$ & $1-R-\eps$ & $\exp(O(1/\eps))$ & randomized & \cite{elias1957list,wozencraft1958list} \\
\hline
Random linear code & $L$ & $\frac{L}{L+1}(1-R-\eps)$ & $\exp(O(L/\eps))$ &randomized & \cite{alrabiah2023randomly} \\
\hline
Random Reed-Solomon & $L$ & $\frac{L}{L+1}(1-R-\eps)$ & $\exp(O(L/\eps)) \cdot n$ & randomized & \cite{alrabiah2023randomly} \\
\hline
Folded AG subcode + hierarchical evasive sets (non-linear code) & $O(1/\eps)$ & $1-R-\eps$ & $\exp(\tilde{O}(1/\eps^2))$ & randomized & \cite{guruswami2022optimal} \\
\hline
AG codes with subfield evaluation + BTT evasive subspace & $2^{(1/\eps)^{O(1)}}$ & $1-R-\eps$ & $\exp(\tilde{O}(1/\eps^2))$ & explicit & \cite{guo2021efficient} \\
\hline
Random AG code & $L$ & $\frac{L}{L+1}(1-R-\eps)$ & $\exp(O(L/\eps))$ &randomized & Thm~\ref{thm:AGpuncture} \\
\hline
\end{tabular}
\caption{Summary of results on nearly optimal list-decodable codes. Here $R$ is the rate and $n$ is the length of the code. For a more comprehensive comparison of prior work, see Table 1 of \cite{kopparty2018improved}.}
\label{tab:listdecoding}
\end{table}

\paragraph{Maximally Recoverable Tensor Codes}
An $(m,n,a,b)$-tensor code is a linear code formed by the tensor product of two codes, an $(n,n-b)$-code $C_{row}$ (called the row code) and an $(m,m-a)$-code $C_{col}$ (called the column code), i.e., $C=C_{col}\otimes C_{row}$. Equivalently, the codewords of $C$ are $m\times n$ matrices whose rows belong $C_{row}$ and columns belong to $C_{col}$. Such a code satisfies `$a$' parity checks per column and `$b$' parity checks per row. Tensor codes have good \emph{locality}, which means that we can recover an erased symbol by reading a small number of remaining symbols (called a repair group). They also have good \emph{availability} which means that there are two such disjoint repair groups, one along the row and one along the column. Thus tensor codes are well-suited in coding for distributed storage. For example, Facebook's (now Meta) f4 storage architecture uses an $(m=3,n=14,a=1,b=4)$ tensor code to store data.

An $(m,n,a,b)$-tensor code is called maximally recoverable (MR) if it can recover from any erasure pattern that is information theoretically possible to recover from (for that particular field characteristic). Thus MR tensor codes are optimal codes in terms of their ability to recover from erasure patterns. The notion of maximal recoverability is introduced by \cite{chen2007maximally,gopalan2012locality} to design optimal codes for distributed storage. Because of their optimality, MR codes are being used in large scale distributed storage  s such as in Microsoft's data centers \cite{huang2012erasure}. MR tensor codes were first studied in the work of \cite{Gopalan2016}, with special emphasis on the case of $a=1$. When $a=1$, the column code is a simple parity check code. In this case, \cite{Gopalan2016} give an explicit condition called \emph{regularity} to check when an erasure pattern is correctable.\footnote{Giving such a condition for general $a,b$ is still open.} \cite{bgm2021mds} showed that an $(m,n,a=1,b)$-tensor code being MR is equivalent to the row code $C_{row}$ being higher order MDS of order $m$, i.e., $\MDS(m)$. Thus MR tensor codes in the regime of $a=1$ are exactly equivalent to higher order MDS codes.

In distributed storage applications, having small field size is extremely important as encoding and decoding involves finite field arithmetic. Therefore constructions of MR tensor codes over small fields is a very important problem. Unfortunately, recent lower bounds due to \cite{bdg2023size} imply that MR tensor codes with $a=1$ and just three rows (i.e., $m=3$) already require exponential field size. Thus it is again natural to look for relaxations of MR tensor codes and hope that we can construct them over smaller fields. In fact, there is a very natural relaxation for MR tensor codes. We say that an $(m,n,a,b)$-tensor code is \emph{$(a',b')$-relaxed MR} if it can correct every erasure pattern that an $(m,n,a',b')$-MR tensor code can correct (here $a'\le a$ and $b'\le b$). In this work, since we are focusing only on MR tensor codes with $a=1$, we will also only look at $(a'=1,b')$-relaxed MR tensor codes. We show that such codes are equivalent to the upper relaxation of higher order MDS codes. We then show that one can indeed construct such codes over small fields by randomly puncturing Reed-Solomon codes or AG codes. Our main result is as follows:

\begin{theorem}{(Informal)}
    \label{thm:relaxedmr_informal}
    Let $C_{col}$ be a simple $(m,m-1)$-parity check code where $m$ is some fixed constant and let $C_{row}$ be an $(n,n-b)$ code sampled as follows for the two different regimes. 
    \begin{enumerate}
        \item Row code has constant rate $R\in (0,1)$, i.e., $b=(1-R)n$: In this case, setting $C_{row}$ to be a randomly punctured AG code over an alphabet of size $\exp(O(m/\eps))$ will make $C_{col}\otimes C_{row}$ into an $(1,(1-\eps) b)$-relaxed $(m,n,1,b)$-MR tensor code with high probability. (Theorem~\ref{thm:AGpunctureMRRate})
        \item Row code has small codimension, i.e., $b\ll n$: In this case, setting $C_{row}$ to be a randomly punctured Reed-Solomon code over an alphabet of constant size $n^{O(m/\eps)}$ will make $C_{col}\otimes C_{row}$ into an $(1,(1-\eps) b)$-relaxed $(m,n,1,b)$-MR tensor code with high probability. (Corollary~\ref{cor:lowCodimMRRS})
    \end{enumerate}
\end{theorem}
The lower bounds from \cite{bdg2023size} imply that in regime (1) of Theorem \ref{thm:relaxedmr_informal}, (non-relaxed) MR tensor codes require fields of size $\exp(\Omega(n))$, whereas the theorem gives constant field size. While in regime (2) of Theorem \ref{thm:relaxedmr_informal}, (non-relaxed) MR tensor codes require fields of size $n^{\Omega(b)}$, whereas the theorem gives polynomial field size independent of codimension $b$.

\subsection{Technical Overview}

We now give an overview of the main techniques needed to prove our main results. We start with motivating the relaxed notions of higher order MDS codes we use to abstract the important properties of our constructions. Then, we discuss how to adapt the state-of-the-art techniques for constructing list-decoding capacity-achieving codes to build more general techniques for constructing relaxed higher order MDS codes.

\subsubsection{Relaxed Higher Order MDS Codes}

As established in \cite{bgm2022}, there is an equivalence between higher order MDS codes, codes with optimal average-case list decoding guarantees, and maximally recoverable tensor codes (in the $a=1$ regime). Thus, in a theory of \emph{relaxed} higher order MDS codes, one would hope that natural relaxations of each of these quantities is also equivalent. However, such an equivalence is not possible (see discussion in Section~\ref{subsec:cmp}). Instead, we define two relaxations of Definition~\ref{def:mds}, which we call the ``lower'' and ``upper'' relaxations of higher order MDS codes.

Consider a $k\times n$ matrix $V$. Informally, we say that $V$ is the \emph{$d$-dimensional lower relaxation} of an $\MDS(\ell)$ code, which we denote by $\rMDS_d(\ell)$, if it behaves like a $(n, k-d)$-$\MDS(\ell)$ code. Likewise we say that $V$ is a \emph{$d$-dimensional upper relaxation} of an $\MDS(\ell)$ code, which we denote by $\rMDS^d(\ell)$, if it behaves like a $(n, k+d)$-$\MDS(\ell)$ code. The precise definition of ``behaves like'' requires a bit of care.

As observed in \cite{tian2019formulas,bgm2021mds}, computing the intersections of spaces is closely related to the following block matrix (see Proposition~\ref{prop:tian-rank}).
\[\mathcal G_{A_1, \hdots, A_\ell}[V] :=  \begin{pmatrix}
I_k & V|_{A_1} & & &\\
I_k & & V|_{A_2} & &\\
\vdots & & & \ddots &\\
I_k & & & & V|_{A_\ell}
\end{pmatrix}.\]
We suppress the sets $A_i$ in the notation if it is clear from context what they are. In fact, an alternative definition of a higher order MDS code is that $\rank \cG[V]$ is always equal to $\rank \cG[W]$, where $W$ is a generic $k\times n$ matrix (see Corollary~\ref{cor:rank-equiv}). Further, one need not check this rank equality for all $A_1, \hdots, A_\ell$, but rather only needs to check when $\cG[W]$ has full column rank (Corollary~\ref{cor:full-column-rank-mds}) or full row rank (Proposition~\ref{prop:MDS-sat}). These alternative definitions of higher order MDS codes lead to our definitions of relaxed MDS codes.

\begin{itemize}
\item We say that a $k \times n$ matrix $V$ is a \emph{lower} relaxation of a higher order MDS code if $\cG_{A_1, \hdots, A_\ell}[V]$ has full column rank whenever $\cG_{A_1, \hdots, A_\ell}[W]$ has full column rank, where $W$ is a generic $(k-d) \times n$ matrix.  (Definition~\ref{def:lower-mds})

\item We say that a $k \times n$ matrix $V$ is an \emph{upper} relaxation of a higher order MDS code if $\cG_{A_1, \hdots, A_\ell}[V]$ has full row rank whenever $\cG_{A_1, \hdots, A_\ell}[W]$ has full row rank, where $W$ is a generic $(k+d) \times n$ matrix.  (Definition~\ref{def:upper-mds})
\end{itemize}

In the process of building the theory of these two relaxations, we prove two equivalence theorems: Theorem~\ref{thm:ldmds-lower-mds} and Theorem~\ref{thm:mds-tensor-equiv}. Theorem~\ref{thm:ldmds-lower-mds} shows that a code $C$ is $\rLDMDS$ if and only if its dual code $C^{\perp}$ is a lower relaxed higher order MDS code. This is a direct generalization of the equivalence theorem of \cite{bgm2022}. Likewise, Theorem~\ref{thm:mds-tensor-equiv} shows that, if $C'$ is a code with a single parity check involving all coordinates, $C' \otimes C$ is an relaxed MR tensor code if and only if $C$ is an upper relaxed higher order MDS code. This generalized an equivalence theorem of \cite{bgm2021mds}.

The proofs of these equivalences are relatively straightforward, and mostly mimic arguments in \cite{bgm2021mds} and \cite{bgm2022}. However, some care is needed in the proofs as we can no longer assume that the matrix $V$ is $\MDS$. The main utility of these equivalencies is that the notions of $\rLDMDS$ and $\rMR$ can be checked using the relatively simple matrix conditions used to check the lower and upper MDS conditions--see Corollary~\ref{cor:slackLDMDSrank} and Definition~\ref{def:saturation}, respectively. This simplicity helps in proving our main results.

\subsubsection{Constructions Using Punctured Codes Coming from Polynomials and Varieties}

The GM-MDS theorem~\cite{yildiz2019gmmds,lovett2018gmmds} using the equivalence in \cite{bgm2022} shows that generic Reed-Solomon (RS) codes are higher order MDS codes (and hence $\LDMDS$ as well). Using this insight \cite{guo2023randomly} and \cite{alrabiah2023randomly} showed that a randomly punctured RS code can achieve list decoding capacity over a linear sized field.

The key idea in \cite{guo2023randomly} was that the average list decoding condition was equivalent to checking the rank of a `Reduced intersection matrix' $M$. As \cite{bgm2022} showed that RS codes achieve list decoding capacity, it was known that $M$ had the right rank for a generic RS code. List decoding close to capacity was then shown to translate to a lower rank condition on $M$. If an RS code was randomly initialized then $M$ not having the right rank would need many `faults' (as generically $M$ had a much higher rank). This leads to a significant advantage in calculating the failure probability and \cite{guo2023randomly} showed that randomly punctured RS codes achieve list decoding capacity over quadratic sized fields. \cite{alrabiah2023randomly} improved this analysis at a few key technical points to get linear field sizes.

Our main result vastly generalizes this and gives results for both the upper and lower relaxations (giving us close to capacity list decodable and relaxed MR-tensor codes). To initialize the \cite{guo2023randomly,alrabiah2023randomly} strategy we need the GM-MDS theorem to hold in a more general setting. In the paper \cite{bdg2023a}, it was shown that the GM-MDS theorem holds for any code where the points are generically from an irreducible variety which contains no hyperplanes through the origin (a particular example would be anything generated by linearly independent polynomials). Using the relaxations discussed earlier we get two matrix conditions for the two relaxations. They generically have much higher rank by the generalized GM-MDS theorem and we are able to follow the argument of \cite{guo2023randomly} and \cite{alrabiah2023randomly} using simple theorems from algebraic geometry (for instance using Bezout's theorem instead of Schwartz-Zippel~\cite{Sch80,Zip79}).

Note, the argument of \cite{guo2023randomly} and \cite{alrabiah2023randomly} cannot be applied directly to the setting of AG codes. Their arguments use the fact that Reed Solomon and Random linear codes are images of polynomial maps. A failure to achieve list decoding capacity implied certain matrices having low rank which gave an algebraic condition on the generator matrix of a code. Substituting the polynomial maps generating the code then allows us to use Schwartz-Zippel to control the number of failures caused by bad columns. AG codes cannot be represented as an image of polynomial maps so we use an algebraic geometric perspective to adapt the arguments of \cite{guo2023randomly} and \cite{alrabiah2023randomly}. We use facts proven using the Reimann-Roch theorem (which is also needed to prove classical 
 distance bounds for AG codes) to show that the columns of an AG code can be treated as points on an irreducible variety of controllable degree. We can then use Bezout's theorem to control the number of bad columns causing an algebraic condition for list-decoding to vanish. This perspective gives us a general statement which works for any code whose columns are sampled from an irreducible variety and recovers the statement for random Reed Solomon and random linear codes from \cite{alrabiah2023randomly}.

\subsection{Open Problems}

We conclude the introduction with a few open problems.

\begin{itemize}
\item One question that remains open is if we can achieve (even existentially) the parameters of the random non-linear code for getting $\eps$-close to list-decoding capacity, i.e., a field size of $\exp(O(1/\eps))$ and list size of $O(1/\eps)$, with a linear code. This would also match the known lower bound of $\exp(\Omega(1/\eps))$ on the field size required to even get polynomial list size \cite{wozencraft1958list,elias1957list}. For non-linear codes $\exp(O(1/\eps))$ is known to be both optimal and attainable~\cite{alrabiah2023ag}. We note that this question is closely related to the current gap in the optimal field size for $(n,k)$-$\MDS(L)$ codes whose currently field size lower/upper bounds are of the form $\exp(\Omega(n))$ and $\exp(O(nL))$, respectively, in the regime that $k/n$ converges to a constant $R \in (0,1)$~\cite{bgm2021mds,bdg2023size}.

\item Do random puncturings of Reed-Solomon or AG codes have efficient list decoding algorithms up to list decoding capacity? As discussed in~\cite{bgm2022}, current hardness results for list decoding do not apply to randomly punctured Reed-Solomon or AG codes.
\item Can we generalize the results in this paper to the duals of punctured AG codes? In particular, does the dual of a punctured AG code achieve list decoding capacity over small fields? Does the dual of a punctured AG code give relaxed MR tensor code with suitable parameters over small fields? Though the generalized GM-MDS theorem of \cite{bdg2023a} applies to duals of polynomial codes, there are further technical difficulties in generalizing the approach in this paper.
\item The upper relaxation of a higher order MDS code corresponds to a relaxed MR$(a,b)$ code in the $a=1$ case. To what extent can we generalize these results for larger $a$? Seemingly a road block is a combinatorial characterization of which patterns are correctable by an $\MR(m,n,a,b)$ tensor code when $a,b\ge 2$. See \cite{bdg2023a} for further discussion. But a more approachable open problem is to show that if $C_{row}$ and $C_{col}$ are both random linear codes of constant rate, then $C_{col}\otimes C_{row}$ is a $((1-\eps)a,(1-\eps)b)$-relaxed $(m,n,a,b)$-MR tensor code over small fields, maybe even fields of size $\exp(O_\eps(1))$ independent of $m,n,a,b$.
\end{itemize}

\subsection*{Organization}

In Section~\ref{sec:prelim} we go over essential background on higher order MDS codes. In Section~\ref{sec:relaxed}, we define various relaxations of higher order MDS codes and show that some of these are equivalent. In Section~\ref{sec:agcode} we show that randomly-punctured algebraic-geometry codes satisfy these relaxed notions over small fields.

\section{Preliminaries}\label{sec:prelim}

In this section, we discuss various background material needed to prove our main results.

\subsection{Notation}

Given a matrix $M \in \F^{m \times n}$, we say that $M$ has \emph{full row rank} if $\rank M = m$ and that $M$ has \emph{full column rank} if $\rank M = n$. When discussing the properties of higher order MDS codes and their relaxations, we need to consider the ranks of various block matrices (e.g.,~\cite{bgm2021mds}). Here, we introduce some succinct notation to discuss these matrices. Given a matrix $V \in \F^{k \times n}$ and sets $A_1, \hdots, A_\ell \subseteq [n]$, we define the following (affine) operators:

\begin{align}\label{eq:Gop}
  \mathcal G_{A_1, \hdots, A_\ell}[V] &:= \begin{pmatrix}
I_k & V|_{A_1} & & &\\
I_k & & V|_{A_2} & &\\
\vdots & & & \ddots &\\
I_k & & & & V|_{A_\ell}
\end{pmatrix},
\end{align}
and
\begin{align}\label{eq:Hop}
  \mathcal H_{A_1, \hdots, A_\ell}[V] &:=  \begin{pmatrix}
    I_n|_{\oA_1} & I_n|_{\oA_2} & \cdots & I_n|_{\oA_\ell}\\
    V|_{\oA_1} & & & \\
     & V|_{\oA_2} & & \\
    & &\ddots & \\
     & & & V|_{\oA_\ell}\\
  \end{pmatrix},
\end{align}
where $\oA_i := [n] \setminus A_i$.

\subsection{Properties of Higher Order MDS Codes}

\begin{proposition}[\cite{tian2019formulas}, as stated in \cite{bdg2023a}]\label{prop:tian-rank}
Let $V$ be a $k \times n$-matrix. For any $A_1, \hdots, A_\ell \subseteq [n]$, we have that
\begin{align}
  \dim(V_{A_1} \cap \cdots \cap V_{A_\ell}) = k + \sum_{i=1}^{\ell} \dim(V_{A_i}) - \rank \cG_{A_1, \hdots, A_\ell}[V].\label{eq:tian}
\end{align}
\end{proposition}

We note the following corollary. 
\begin{corollary}\label{cor:rank-equiv}
Let $V$ be a $k \times n$-matrix and $W$ a generic $k \times n$-matrix. For all $\ell \ge 2$, we have that $V$ is $\MDS(\ell)$ if and only if for all $A_1, \hdots, A_\ell \subseteq [n]$, 
\begin{align}
\rank \mathcal G_{A_1, \hdots, A_\ell}[V] = \rank \cG_{A_1, \hdots, A_\ell}[W] \label{eq:V-W}.
\end{align}
\end{corollary}

As we believe this specific observation hasn't been noted before in the literature, we include a proof.

\begin{proof}
  If $V$ is $\MDS(\ell)$, then (\ref{eq:V-W}) follows immediately from Proposition~\ref{prop:tian-rank}.

  Conversely, assume (\ref{eq:V-W}) holds for all $A_1, \hdots, A_\ell \subseteq [n]$. To show that $\dim(V_{A_1} \cap \cdots \cap V_{A_\ell}) = \dim(W_{A_1} \cap \cdots \cap W_{A_\ell})$, it suffices by Proposition~\ref{prop:tian-rank} to prove that $\dim(V_{A_i}) = \dim(W_{A_i})$ for all $i \in [\ell]$. To prove this, consider any $A \subseteq [n]$, if we set $A_1 = \cdots = A_\ell = A$ and apply Proposition~\ref{prop:tian-rank} with (\ref{eq:V-W}), we get that $\dim(V_{A}) = \dim(W_A)$ for all $A \subseteq [n]$, as desired. 
\end{proof}

We now state what we believe is a novel adaptation of Proposition~\ref{prop:tian-rank} for parity-check matrices. The proof is in Appendix~\ref{app:omit}. Note that this formula also relates the rank of $\cG(G)$ to the rank of $\cH(G^{\perp})$.

\begin{restatable}{lemma}{lemdual}\label{lem:dual-inter-rank}
 Let $C$ be an $(n,k)$-code (not necessarily MDS) with generator matrix $G$ and parity-check matrix $H$. For any $A_1, \hdots, A_\ell \subseteq [n]$, we have that
\begin{align}
  \dim(H_{A_1} \cap \cdots \cap H_{A_\ell}) = n - k +  \sum_{i=1}^{\ell} \rank(G|_{\oA_i}) - \rank \cH_{A_1,\hdots,A_\ell}[G].\label{eq:dim-dual}
\end{align}
\end{restatable}
\subsubsection{Null Intersection}

An important structural property of higher order MDS codes is characterizing when the intersection of spaces is the null space.

\begin{definition}[\cite{bgm2021mds}]
  We say that $A_1, \hdots, A_\ell \subseteq [n]$ have the $k$-dimensional null intersection property if for any generic $k\times n$ matrix $W$, we have that $W_{A_1} \cap \cdots \cap W_{A_\ell} = 0$.
\end{definition}

Combinatorial characterizations of the null intersection property were found in \cite{bgm2021mds,bgm2022}.

\begin{proposition}[\cite{bgm2022}]\label{prop:null}
  Consider sets $A_1, \hdots, A_\ell \subseteq [n]$ with $|A_i| \le k$ for all $i \in [n]$. These sets have the $k$-dimensional null intersection property if and only if for all partitions $P_1 \cup \cdots \cup P_s = [\ell]$, we have that
  \begin{align}
    \sum_{i=1}^{s} \left|\bigcap_{j \in P_i} A_j\right| \le (s-1)k \label{eq:part-sum}
  \end{align}
\end{proposition}

We note that it suffices to verify Definition~\ref{def:mds} for $A_1, \hdots, A_\ell$ has the $k$-dimensional null intersection property.

\begin{proposition}[\cite{bgm2021mds}]\label{prop:2.2}
For $\ell \ge 2$, a generator matrix $V \in \F^{k \times n}$ with every column nonzero is $\MDS(\ell)$ if and only if for all $A_1, \hdots, A_\ell \subseteq [n]$ with the $k$-dimensional null intersection property and $|A_i| \le k$ we have $V_{A_1} \cap \cdots \cap V_{A_\ell} = 0$.
\end{proposition}

We also note the following simple observation will help motivate the lower relaxation of higher order MDS codes.

\begin{corollary}\label{cor:full-column-rank-mds}
Let $V$ be a $k \times n$-matrix and $W$ a generic $k \times n$-matrix. For all $\ell \ge 2$, we have that $V$ is $\MDS(\ell)$ if and only if for all $A_1, \hdots, A_\ell \subseteq [n]$, such that $\cG_{A_1, \hdots, A_\ell}[W]$ has full column rank, we have that (\ref{eq:V-W}) holds.
\end{corollary}

\begin{proof}
  The ``only if'' direction follows from Corollary~\ref{cor:rank-equiv}. For the ``if'' direction, consider any $A_1, \hdots, A_\ell \subseteq [n]$ with $|A_i| \le k$ with the $k$-dimensional null intersection property. Note by Proposition~\ref{prop:tian-rank} that $\cG_{A_1, \hdots, A_\ell}[W]$ has full column rank. Thus, by (\ref{eq:V-W}), we have that $\cG_{A_1, \hdots, A_\ell}[V]$ has full column rank. By Proposition~\ref{prop:tian-rank}, this can only happen if $V_{A_1} \cap \cdots \cap V_{A_\ell} = 0$. Thus, $V$ is $\MDS(\ell)$ by Proposition~\ref{prop:2.2}.
\end{proof}

\subsubsection{Saturation}

As noted, the null intersection property is closely related to $\cG$ having full column rank. We now consider the ``dual'' situation in which $\cG$ has full row rank.  We call this the \emph{saturation property}. This property is useful for motivating the upper relaxation of higher order MDS codes.

\begin{definition}\label{def:saturation}
Let $V \in \F^{k \times n}$.  We say that $A_1, \hdots, A_\ell \subseteq [n]$ are \emph{$V$-saturated} if
\[
  \rank \mathcal G_{A_1, \hdots, A_\ell}[V] = \ell k,
\]
i.e., $\cG_{A_1, \hdots, A_\ell}[V]$ has full row rank. We say the sets have the \emph{$k$-dimensional saturation property} if they are $W$-saturated, where $W$ is a $k\times n$ generic matrix.
\end{definition}

By Proposition~\ref{prop:tian-rank}, we have that $A_1, \hdots, A_\ell$ are $V$-saturated if and only if
\begin{align}
  \dim(V_{A_1} \cap \cdots \cap V_{A_\ell}) = \sum_{i=1}^{\ell} \dim(V_{A_i}) - (\ell-1)k.\label{eq:sat-dim}
\end{align}

We note that Proposition~\ref{prop:2.2} for null intersection has an analogous result for saturation.

\begin{proposition}\label{prop:MDS-sat}
For $\ell \ge 2$, a generator matrix $V \in \F^{k \times n}$ is $\MDS(\ell)$ if and only if for all $A_1, \hdots, A_\ell \subseteq [n]$ with the $k$-dimension saturation property we have that $A_1, \hdots, A_\ell$ are $V$-saturated.
\end{proposition}

\begin{proof}
First, if $V$ is $\MDS(\ell)$, then for any $A_1, \hdots, A_\ell \subseteq [n]$ with the $k$-dimension saturation property we have by Corollary~\ref{cor:rank-equiv} that $A_1, \hdots, A_\ell$ are $V$-saturated.

Conversely, for all $A_1, \hdots, A_\ell \subseteq [n]$ with the $k$-dimension saturation property we have that $A_1, \hdots, A_\ell$ are $V$-saturated. By Corollary~\ref{cor:rank-equiv}, we seek to show that (\ref{eq:V-W}) holds for all $A_1, \hdots, A_\ell \subseteq [n]$. Let $r = \rank \cG_{A_1, \hdots, A_\ell}[W]$ and observe that $\ell k = \rank \cG_{[n], \hdots, [n]}[W]$. Thus, there exist $A'_1, \hdots, A'_\ell \subseteq [n]$ with $\sum_{i=1}^{\ell} |A'_i \setminus A_i| = \ell k - r$ such that $\ell k = \rank \cG_{A'_1, \hdots, A'_\ell}[W]$. Hence, $A'_1, \hdots, A'_\ell$ have the $k$-dimensional saturation property.

Thus, by assumption, we have that $\ell k = \rank \cG_{A'_1, \hdots, A'_\ell}[V]$, so $\rank \cG_{A_1, \hdots, A_\ell}[V] \ge \ell k - (\ell k - r) = r$. Since a generic matrix maximizes rank, we have that $\rank \cG_{A_1, \hdots, A_\ell}[V] = r$, as desired.
\end{proof}

\subsubsection{MR Tensor Codes}

Let $\mathcal E_{a,b,p}^{m,n}$ be the set of all subsets $E \subseteq [m] \times [n]$ which are correctable by a $(m, n, a, b)$-MR-tensor code over a field of characteristic $p$. If the base field is clear from context, we drop mention of $p$. It is easy to see that this family is monotone in $a$ and $b$: for all $a' \le a$ and $b' \le b$, we have that $\mathcal E_{a',b'}^{m,n} \subseteq \mathcal E_{a,b}^{m,n}$. 

We also let $\mathcal P_{a,b}^{m,n} \subset \mathcal E_{a,b}^{m,n}$ be the set of maximal patterns. In particular, if $m \ge a$ and $n \ge b$, then $\mathcal P_{a,b}^{m,n} = \{E \in \mathcal E_{a,b}^{m,n} : |E| = bm+an-ab\}$. We recall the regularity theorem of \cite{Gopalan2016}.

\begin{theorem}[\cite{Gopalan2016}]\label{thm:regularity}
  For $E \subseteq [m] \times [n]$, we have that $E \in \mathcal E_{a,b}^{m,n}$ only if for all $S \subseteq [m]$ with $|S| \ge a$ and $T \subseteq [n]$ with $|T| \ge b$, we have that
  \[
    |E \cap (S \times T)| \ge b|S| + a|T| - ab.
  \]
  Further, this description of $\mathcal E_{a,b}^{m,n}$ is complete if $a=1$.
\end{theorem}

The following restate some results in \cite{bgm2021mds}, their proofs are in Appendix~\ref{app:omit}. Correctable erasure patterns in the case of $a=1$ are related to sets with the saturation property.

\begin{restatable}{proposition}{propmrtensor}\label{prop:mr-tensor}
Let $C_1$ be a $(m,m-a)$-code with generator matrix $U$. Likewise, let $C_2$ be an $(n,n-b)$-code with generator matrix $V$. Let $E \subseteq [m] \times [n]$. Also let $A_1, \hdots, A_m \subseteq [n]$ such that $\bar{E} = \bigcup_{i=1}^{m} \{i\} \times A_i$ the following are equivalent:

\begin{enumerate}
\item[(a)] $E$ is a correctable pattern.
\item[(b)] $\sum_{(i,j) \in \bar{E}} \langle U_i \otimes V_j\rangle = \F^{m-a} \otimes \F^{n-b}.$
\item[(c)] If $a = 1$ and $C_1$ is MDS, $\rank \mathcal G_{A_1, \hdots, A_m}[V] = m(n-b)$, i.e., $A_1, \hdots, A_m$ are $V$-saturated.
\end{enumerate}
\end{restatable}

\begin{remark}
  Likewise, from the results of \cite{bgm2021mds}, one can prove that the vectors $\{U_i \otimes V_j : (i,j) \in \bar{E}\}$ are independent if and only if $G_{A_1, \hdots, A_m}[V]$ has full column rank, i.e., $V_{A_1} \cap \hdots \cap V_{A_\ell} = 0$ and $\dim V_{A_i} = |A_i|$ for all $i \in [m]$.

  Thus, sets with the null intersection property are related to independent sets in the tensor product matroid, i.e., $A_1,A_2,\dots,A_m$ have $k$-dimensional null-intersection property if generically $\{U_i \otimes V_j : j\in A_i\}$ are independent in $\F^{m-1} \otimes \F^{k}$. Likewise, sets with the saturation property are related to spanning sets in the tensor product matroid, i.e., $A_1,A_2,\dots,A_m$ have $k$-dimensional saturation property if $\{U_i \otimes V_j : j\in A_i\}$ span $\F^{m-1} \otimes \F^{k}$.
\end{remark}

\begin{restatable}[Implicit in \cite{bgm2021mds}]{lemma}{lemnoOfMRpatterns}\label{lem:noOfMRpatterns}
  Let $C_1 \subseteq \F^m$ and $C_2 \subseteq \F^n$ be codes with distance at least $a+1$ and $b+1$ respectively. Then, there exists a set $\mathcal P_{a,b}^{m,n} \subseteq \mathcal E^{m,n}_{a,b}$ such that $C_1 \otimes C_2$ is $(a,b)$-$\MR$ if and only if for all $P \in \mathcal P_{a,b}^{m,n}$ we have that $C_1 \otimes C_2$ can correct erasure pattern $P$ and
\[
  |\mathcal P_{a,b}^{m,n}| \le \min\left[2^{mn}, \binom{n}{\le b(m-a)} \binom{bm(m-a)}{\le b(a+1)(m-a)}, \binom{m}{\le a(n-b)} \binom{an(n-b)}{\le a(b+1)(n-b)}\right].
\]
\end{restatable}

\section{Relaxed Higher Order MDS Codes}\label{sec:relaxed}

In this section, we define lower and upper relaxations of higher order MDS codes. We then show that lower relaxation is equivalent (up to duality) to relaxed version of optimal list-decodable codes ($\LDMDS$), and upper relaxation is equivalent to relaxed version of MR tensor codes (with single column parity check).

\subsection{Lower and Upper Relaxations of Higher Order MDS Codes}

As discussed in Section~\ref{sec:prelim}, there are many equivalent criteria for checking if a code is a higher order MDS code, including Corollary~\ref{cor:full-column-rank-mds} and Proposition~\ref{prop:MDS-sat}. Each of these criteria is a logical conjunction of many algebraic conditions (e.g., a matrix is of a specific rank, etc.). As such, these formulations of a higher order MDS code can be relaxed by only imposing that a (suitable) subset of the algebraic conditions hold. The precise relaxation needs to be carefully chosen so that many ``natural'' properties of higher order MDS codes continue to hold in these relaxed versions.

More precisely, our \emph{lower} relaxation (Definition~\ref{def:lower-mds}) relaxes Corollary~\ref{cor:full-column-rank-mds} so that our code looks like a higher order MDS code of \emph{smaller} rate; while our \emph{upper} relaxation (Definition~\ref{def:upper-mds}) relaxes Proposition~\ref{prop:MDS-sat} to look like a higher order MDS code of \emph{larger} rate.

\subsubsection{Lower Relaxation} The lower relaxation checks the condition of Corollary~\ref{cor:full-column-rank-mds} on a smaller number of families of sets.

\begin{definition}[Lower Relaxation]\label{def:lower-mds}
  Let $\ell \ge 2$ and $n \ge k \ge d \ge 0$. Let $W$ be a $(k-d) \times n$ generic matrix. We say that a $(n,k)$-code $C$ with generator matrix $G$ is $\rMDS_d(\ell)$ if $\mathcal G_{A_1, \hdots, A_\ell}[G]$ has full column rank whenever $\mathcal G_{A_1, \hdots, A_\ell}[W]$ has full column rank.
\end{definition}

We now note an equivalent definition, whose equivalence we prove in Appendix~\ref{app:omit}.

\begin{restatable}{proposition}{propA}\label{prop:lower-mds-alt}
  Let $\ell \ge 2$ and $n \ge k \ge d \ge 0$. We say that a $(n,k)$-code $C$ with generator matrix $G$ is $\rMDS_d(\ell)$ if and only if (1) every column of $G$ is nonzero and (2) for all $A_1, \hdots, A_\ell \subseteq [n]$ of size at most $k-d$ with the $k-d$-dimensional null intersection property we have that $G_{A_1} \cap \cdots \cap G_{A_\ell} = 0.$
\end{restatable}

We now state a few basic properties of this relaxation, many of which are analogous to results in \cite{bgm2021mds,roth2021higher}. These are proved in Appendix~\ref{app:omit}.

\begin{restatable}{proposition}{propB}\label{prop:lower-mds}
Let $C$ be an $(n,k)$-code with generator matrix $G$. Let $\ell \ge 2$.
\begin{itemize}
\item[(a)] $C$ is $\rMDS_0(\ell)$ if and only if $C$ is $\MDS(\ell)$.  
\item[(b)] If $C$ is $\rMDS_d(\ell)$, then $C$ is $\rMDS_{d'}(\ell)$ for all $d' \in \{d+1, \hdots, k\}$.
\item[(c)] If $C$ is $\rMDS_d(\ell)$, then $C$ is $\rMDS_d(\ell')$ for all $\ell' \in \{2, \hdots, \ell-1\}$.
\item[(d)] $C$ is $\rMDS_d(2)$ if and only if every $k-d$ columns of $G$ are linearly independent. 
\end{itemize}
\end{restatable}

\subsubsection{Upper Relaxation}

We now define the upper relaxation of a higher order MDS code based on the saturation property (Definition~\ref{def:saturation}).

\begin{definition}[Upper Relaxation]\label{def:upper-mds}
Let $\ell \ge 2$. We say that a $(n,k)$-code $C$ with generator matrix $G$ is $\rMDS^d(\ell)$ with $0 \le d \le n - k$ if for all $A_1, \hdots, A_\ell \subseteq [n]$ with the $k+d$-dimensional saturation property, we have that $A_1, \hdots, A_\ell$ are $G$-saturated.
\end{definition}

We note that unlike the lower relaxation, some columns of $G$ may equal zero. The following properties are proved in Appendix~\ref{app:omit}.

\begin{restatable}{proposition}{propC}\label{prop:upper-mds}
Let $C$ be an $(n,k)$-code with generator matrix $G$. Let $\ell \ge 2$.
\begin{itemize}
\item[(a)] $C$ is $\rMDS^0(\ell)$ if and only if $C$ is $\MDS(\ell)$.  
\item[(b)] If $C$ is $\rMDS^d(\ell)$, then $C$ is $\rMDS^{d'}(\ell)$ for all $d' \ge d+1$.
\item[(c)] If $C$ is $\rMDS^d(\ell)$, then $C$ is $\rMDS^d(\ell')$ for all $\ell' \in \{2, \hdots, \ell-1\}$.
\item[(d)] $C$ is $\rMDS^d(2)$ if and only if every $k+d$ columns of $G$ span $\F^k$.
\end{itemize}
\end{restatable}

\subsubsection{Comparison Between Relaxations}\label{subsec:cmp}

A natural question is if the distinction between upper and lower relaxations is truly necessary. Note that one cannot directly compare the upper and lower relaxations: if $C$ is an $(n,k)$-code. There are $k+1$ values of $d$ for which one can ask whether $C$ is an $\rMDS_d(\ell)$ code, but there are $n-k+1$ values of $d$ for which one can ask whether $C$ is an $\rMDS^d(\ell)$ code.
For this comparison to ``type check'' we need to compare whether $C$ is $\rMDS_d(\ell)$ with whether the dual code $C^{\perp}$ is $\rMDS^d(\ell)$. From this perspective, these notions are equivalent for $\ell = 2$, but not necessarily for large $\ell$. We first demonstrate the equivalence part of this statement.

\begin{proposition}
  Let $C$ be an $(n,k)$-code. We have that $C$ is $\rMDS_d(2)$ if and only if $C^{\perp}$ is $\rMDS^d(2)$.
\end{proposition}

\begin{proof}
  Let $G \in \F^{k \times n}$ be a generator matrix of $C$ and let $H \in \F^{(n-k) \times n}$ be a parity-check matrix of $C$. By Proposition~\ref{prop:lower-mds}(d), we have that $C$ is $\rMDS_d(2)$ if and only if every $k-d$ columns of $G$ are linearly independent. Likewise, by Proposition~\ref{prop:upper-mds}(d), we have that $C^{\perp}$ is $\rMDS^d(2)$ if and only if every $n-k+d$ columns of $H$ span $\F^k$.

To see why these are equivalent, first note if $k=d$, these equivalences are vacuous. Otherwise, fix any $S \subseteq [n]$ of size $k-d$ and consider the following chain of equivalences:
\begin{itemize}
\item The columns of $G|_S$ are linearly dependent.
\item There exists $S \subseteq [n]$ of size at most $k-d$ and nonzero $x \in \F^n$ with $\supp x \subseteq S$ such that $Gx = 0$.
\item There exists nonzero $x \in \F^n$ supported by $S$ and nonzero $y \in \F^k$ such that $x = y^{\perp} H$.
\item There exists nonzero $y \in \F^k$ such that for every vector $z \in H_{\bar S}$, we have that $y^{\perp} z \neq 0$.
\item The columns of $H|_{\bar S}$ fail to span $\F^k$.
\end{itemize}
This completes the proof.
\end{proof}

However, for larger $\ell$ this equivalence breaks down.

\begin{proposition} The following statements hold.
  \begin{itemize}
    \item[(a)] There exists an $\rMDS_d(4)$ code $C$ such that $C^{\perp}$ fails to be $\rMDS^d(4).$
    \item[(b)] There exists an $\rMDS^d(4)$ code $C$ such that $C^{\perp}$ fails to be $\rMDS_d(4)$.
  \end{itemize}
\end{proposition}

\begin{proof}
  These equivalences break down due to counterexamples appeared in \cite[Appendix B.3]{bgm2021mds}.
\end{proof}

\subsection{Relaxed LD-MDS Codes}

To understand codes achieving a list-decoding capacity over small fields, Guo and Zhang~\cite{guo2023randomly} used a relaxed Singleton bound. We abstract that into a definition.

\begin{definition}[Relaxed LD-MDS Codes]\label{def:relaxed-ld-mds}
Given $d \ge 0$, let $\rho_{n,k,d}(L) := \frac{L}{L+1}\frac{n-k-d}{n}$. We say that a $(n,k)$-code $C$ is $\rLDMDS_d(L)$ if it is $(\rho_{n,k,d}(L), L)$ average-radius list-decodable.  We say a code is $\rLDMDS_d(\le L)$ if it is $\rLDMDS_d(\ell)$ for all $\ell \in \{1, \hdots, L\}$.
\end{definition}

Note that $\rLDMDS_0(L)$ coincides with $\LDMDS(L)$. However, when $d \ge 1$, a $\rLDMDS_d(1)$ code may not be MDS. However, we can prove such a code has large distance (analogous to a result of \cite{roth2021higher}).

\begin{proposition}\label{prop:rldmds-1}
  Let $C$ be an $(n,k)$-code. Then, $C$ is $\rLDMDS_d(1)$ if and only if its distance is at least $n-k-d+1$.
\end{proposition}

\begin{proof}
  First, assume $C$ is $\rLDMDS_d(1)$. If $\dim C = 0$, the result is vacuous. Otherwise, consider any pair $c_1, c_2 \in C$. Since $C$ is $(\rho_{n,k,d}(1), 1)$ average-radius list-decodable, we have for $y = c_1$ that $\wt(c_1 - c_2) > n-k-d$, as desired.

  Second, assume that $C$ has distance at least $n-k-d+1$. For any $c_1, c_2 \in C$ and any $y \in \F^n$, we have that
\[
\wt(c_1 - y) + \wt(c_2 - y) \ge \wt(c_1 - c_2) \ge n-k-d+1.
\]
Thus, $C$ is $(\rho_{n,k,d}(1), 1)$ average-radius list-decodable, as desired.
\end{proof}
As a corollary, we have that $\rLDMDS_d(1)$ is dual to $\rMDS_d(2)$.

\begin{corollary}\label{cor:dual-ldmds-1}
  Let $C$ be an $(n,k)$-code. Then, $C$ is $\rLDMDS_d(1)$ if and only if $C^{\perp}$ is $\rMDS_d(2)$.
\end{corollary}

\begin{proof}
  Let $G$ be a generator matrix of $C$, and let $H$ be a parity-check matrix of $C$.  By Proposition~\ref{prop:rldmds-1} and Proposition~\ref{prop:lower-mds}, it suffices to prove that $C$ has distance at least $n-k-d+1$ if and only if every $n-k-d$ columns of $H$ are linearly independent. To see this equivalence, we have that the columns $A \subseteq [n]$ of $H$ are linearly independent if and only if there is no $c \in \F^n$ with $\supp c \subseteq A$ and $Hc = 0$, or equivalently, $c \in C$. Thus, every $n-k-d$ columns of $H$ are linearly independent if and only if every nonzero codeword of $C$ has Hamming weight at least $n-k-d+1$, as desired.
\end{proof}

We now generalize this duality for all values of $L$. In particular, our result is a generalization to relaxed higher order MDS codes of the main theorem from \cite{bgm2022} that $\LDMDS(\le L)$ is dual to $\MDS(L+1)$.

\begin{theorem}\label{thm:ldmds-lower-mds}
  Let $C$ be an $(n,k)$-code. For any $d \in \{0, \hdots, n-k\}$, we have that $C$ is $\rLDMDS_d(\le L)$ if and only if $C^{\perp}$ is $\rMDS_d(L+1)$.
\end{theorem}

\begin{proof}
  We prove both directions via contrapositive. Let $H$ be the parity-check matrix of $C$.

  First, assume $C$ is $\rLDMDS_d(\le L)$ but $C^{\perp}$ is not $\rMDS_d(L+1)$. Thus, there exists $A_1, \hdots, A_{L+1} \subseteq [n]$ with the $n-k-d$-dimensional null intersection property, but there exists $z \neq 0$ such that $z \in H_{A_1} \cap \cdots \cap H_{A_{L+1}}$.  Thus, there are $x_1, \hdots, x_{L+1} \in \F^{n}$ such that $Hx_i = z$ and $\supp x_i \subseteq A_i$. Let $c_i := x_i - x_1 \in C$. Then, $A_i \supseteq \supp x_i = \supp (c_i - (-x_1)).$ Let $P_1, \hdots, P_s$ be a partition of $[{L+1}]$ such that $j,j' \in P_i$ if and only if $c_j = c_{j'}$ (i.e., $x_j = x_{j'}$). For each $i \in [s]$ and $j \in P_i$, we have that $\supp x_j \subseteq \bigcap_{j' \in P_i} A_{j'}$.

If $s = 1$, then by Proposition~\ref{prop:null}, we have that $\supp x_j = \emptyset$ for all $j \in [{L+1}]$. This contradicts that $Hx_j = z \neq 0$. Thus, $s \ge 2$. We can thus pick representatives $j_i \in P_i$ for all $i \in [s]$, and observe that
\[
  \sum_{i=1}^{s} \wt(c_{j_i} - (-x_1)) \le \sum_{i=1}^s \left|\bigcap_{j \in P_i} A_j\right| \le (s-1) (n - k - d).
\]
Thus,  $c_{j_1}, \hdots, c_{j_s}$ are distinct points in a ball with center $-x_1$ of average radius $\rho_{n,k,d}(s-1)$, a contradiction of the fact that $C$ is $\rLDMDS_d(\le L)$.

Second, assume that $C^{\perp}$ is $\rMDS_d(L+1)$ but $C$ is not $\rLDMDS_d(L')$ for some $L' \in \{1, \hdots, L\}$. Assume further that the choice of $L'$ is minimal. By Proposition~\ref{prop:lower-mds} and Corollary~\ref{cor:dual-ldmds-1}, we have that $C$ is $\rLDMDS_d(1)$. Therefore, $L' \ge 2$ and $C$ has distance at least $n-k-d+1$ by Proposition~\ref{prop:rldmds-1}. 

Hence, there are distinct $c_1, \hdots, c_{L'+1} \in C$ and $y \in \F^n$ such that if $A_i := \supp (c_i - y)$, then \begin{align}|A_1| + \cdots + |A_{L'+1}| \le L'(n - k - d).\label{eq:sum-ldmds}\end{align} If some $i \in A_1 \cap \cdots \cap A_{L'+1}$, then we can pick $y_i = (c_1)_i$ without increasing the size of any $A_i$. Thus, we may WLOG that $A_1 \cap \cdots \cap A_{L'+1} = \emptyset$. Also observe that $Hy \in H_{A_1} \cap \cdots \cap H_{A_{L'+1}}$. We now break the proof into three cases.

\begin{itemize}
\item[(a)] $Hy = 0$.
\item[(b)] $A_1, \hdots, A_{L'+1}$ satisfy the $n-k-d$-dimensional null intersection property. 
\item[(c)] Otherwise.
\end{itemize}

\paragraph{Case (a).}
 If $Hy = 0$, then $y$ is a codeword, so if $y \neq c_i$, then $|A_i| \ge n - k - d + 1$. Therefore, $\sum_{i=1}^{L'+1} |A_i| \ge L'(n - k - d + 1)$, a contradiction.

\paragraph{Case (b).}
Since $Hy \neq 0$ and $Hy \in H_{A_1} \cap \cdots \cap H_{A_{L'+1}}$, we must have that $C^{\perp}$ is not $\rMDS_d(L'+1)$, a contradiction (by Proposition~\ref{prop:lower-mds}).

\paragraph{Case (c).} By Proposition~\ref{prop:null} there must exist partition $P_1, \hdots, P_s$ of $[L'+1]$ for which
\begin{align}
  \sum_{i=1}^{s} \left|\bigcap_{j \in P_i} A_j\right| > (s-1)(n-k-d) \label{eq:532}
\end{align}
 with $s$ maximal. Since $A_1 \cap \cdots \cap A_\ell = \emptyset$, we may assume that $s \ge 2$. Subtracting (\ref{eq:sum-ldmds}) from (\ref{eq:532}), we then have that
\[
  \sum_{i=1}^{s} \left[\sum_{j \in P_i} |A_j| - \left|\bigcap_{j \in P_i} A_j\right|\right] < (L' + 1 - s)(n - k - d).
\]
By the pigeonhole principle, there is $i \in [s]$ such that
\[
   \sum_{j \in P_i} |A_j| - \left|\bigcap_{j \in P_i} A_j\right| <  (|P_i| - 1)(n - k(\lambda + 1)).
\]
Both the LHS and RHS are $0$ when $|P_i| = 1$, so $|P_i| \ge 2$. Further, since $s \ge 2$, we have that $|P_i| \le L'$.

Let $S := \bigcap_{j \in P_i} A_j$. Let $z \in \F^n$ be such that $z|_S = (x_{j_0})|_S$ for some arbitrary $j_0 \in P_i$, but $z|_{\bar{S}} = y|_{\bar{S}}$. Then, since $S \subseteq A_j$ for all $j \in P_i$, we have that
\begin{align*}
  \sum_{j \in P_i}\wt(c_j - z) &= \sum_{j \in P_i \setminus \{j_0\}} |A_j| + |A_{j_0} \setminus S|\\
  &= \sum_{j \in P_i} |A_j| - \left|\bigcap_{j \in P_i} A_j\right|\\
  &<  (|P_i| - 1)(n - k(\lambda + 1)).
\end{align*}
Therefore,  $C$ is not $\rLDMDS_d(|P_i| - 1)$. This contradicts that $L'$ is minimal. Therefore, $C$ is $\rLDMDS_d(\le L)$, as desired.
\end{proof}

The follow corollary will be useful for analyzing the list decoding properties of algebraic-geometry codes. This result is essentially equivalent to Lemma 3.5 of \cite{guo2023randomly}.

\begin{corollary}\label{cor:slackLDMDSrank}
 Let $C$ be a $(n,k)$-code. We have that $C$ is $\LDMDS_d(\le L)$ if and only if for all $\ell \in \{2, \hdots, L+1\}$, and for all $A_1, \hdots A_\ell \subseteq [n]$ with the $n-k-d$-dimensional null intersection property we have that
  \begin{align}
   \rank \cH_{A_1, \hdots, A_\ell}[G] = n+(\ell-1)k. \label{eq:643}
  \end{align}
\end{corollary}
\begin{proof}

By Theorem~\ref{thm:ldmds-lower-mds}, we have that $C$ is $\rLDMDS_d(\le L)$ if and only if for all $\ell \in \{2, \hdots, L+1\}$, and for all $A_1, \hdots A_\ell \subseteq [n]$ with the $n-k-d$-dimensional null intersection property,  we have that $H_{A_1} \cap \cdots \cap H_{A_\ell} = 0,$ where $H$ is the parity-check matrix of $C$. Assume the distance of $C$ is at least $n - k - d + 1$. Therefore, no codeword of $C$ can have support inside $A_i$. Thus, the map $x \mapsto x|_{\oA_i}$ is an injection for $C$. Thus, $\rank(G|_{\oA_i}) = k$ for all $i \in [\ell]$. Applying Lemma~\ref{lem:dual-inter-rank}, we have that $H_{A_1} \cap \cdots \cap H_{A_\ell} = 0$ is equivalent to (\ref{eq:643}), as desired.

To prove the distance assumption, if $C$ is $\rLDMDS_d(\le L)$, then $C$ has distance at least $n - k - d+1$ from Proposition~\ref{prop:rldmds-1}. Conversely, if (\ref{eq:643}) holds for all specified tuples of sets, consider $\ell = 2$, $A_1$ be any set of size $n-k-d$ and $A_2 = \emptyset$. Applying Lemma~\ref{lem:dual-inter-rank}, we must have that \[
0 \le \dim(H_{A_1} \cap \cdots \cap H_{A_\ell}) = n - k + \rank(G|_{\oA_1}) + \rank(G) - (n+k) = \rank(G|_{\oA_1}) - k.
\]
Thus, $\rank(G|_{\oA_1}) = k$, so nonzero codeword of $C$ can be supported on $A_1$, as desired.
\end{proof}

For applications to list decoding, we need to understand how (\ref{eq:643}) changes when some columns of $G$ are deleted.

\begin{proposition}\label{prop:LDMDSrank-puncture}
 Let $C$ be a $(n,k)$-code with generator matrix $G$. Let $A_1, \hdots A_\ell \subseteq [n]$. For all $B \subseteq \oA_1 \cup \cdots \oA_\ell$, we have that
  \begin{align}
   \rank \cH_{A_1, \hdots, A_\ell}[G] \ge |B| + \rank \cH_{A_1\setminus B, \hdots, A_\ell\setminus B}[G|_{\oB}]. \label{eq:644}
  \end{align}
\end{proposition}

\begin{proof}
  Let $M$ be the matrix on the LHS of (\ref{eq:644}) and $M_B$ be the matrix on the RHS of (\ref{eq:644}). Let $r = \rank M_B$. In particular, there must exist an $r \times r$ submatrix $X$ of $M_B$ with nonzero determinant.
For each $i \in B$, let $j_i \in [\ell]$ be such that $i \in \oA_{j_i}$ (which must exist by the hypothesis for $B$).

Let $Y$ be the $(r + |B|) \times (r + |B|)$ submatrix of $M$ such that $r$ rows and $r$ columns correspond to $X$. The additional $|B|$ rows correspond to the rows of the identity row block corresponding to $B$ and the additional $|B|$ columns correspond to $G_i$ in the $G|_{\oA_{j_i}}$ block for $i \in B$. Note that $\det Y = \pm \det X \neq 0$, as expanding the determinant along the $|B|$ new rows of $Y$ just leave a single term corresponding to $X$. Thus, $\rank M \ge |B| + \rank M_B$.
\end{proof}
\subsection{Relaxed MR Tensor Codes}

\begin{definition}[Relaxed MR Tensor Codes]\label{def:mr-tensor-codes}
  Let $0 \le a' \le a \le m$ and $0 \le b' \le b \le n$. Let $C_{col}$ be a $(m,m-a)$-code and $C_{row}$ be a $(n,n-b)$-code. We say that $C_{col} \otimes C_{row}$ is a $(a',b')$-relaxed $(m,n,a,b)$-MR tensor code if $C_{col} \otimes C_{row}$ can correct any erasure pattern $E \in \mathcal E_{a',b'}^{m,n}$.
\end{definition}

Unlike for MR-tensor codes, $C_{col}$ and $C_{row}$ need not be MDS codes. We show for $a=1$ this condition can be made equivalent to a suitable upper relaxation of a higher order MDS code.

\begin{theorem}\label{thm:mds-tensor-equiv}
  Let $C$ be an $(n,n-b)$-code. Let $d \in \{0, 1, \hdots, b\}$. Let $C'$ be any $(m,m-1)$-MDS-code. The following are equivalent
  \begin{itemize}
    \item[(a)] $C$ is $\MDS^d(m)$.
    \item[(b)] $C' \otimes C$ is $(1,b-d)$-$\rMR(m,n,1,b)$.
  \end{itemize}
\end{theorem}

We first observe the following corollary of Proposition~\ref{prop:mr-tensor}

\begin{proposition}\label{prop:sat-equiv}
  Consider $E \subseteq [m] \times [n]$. Let $A_1, \hdots, A_m \subseteq [n]$ be such that $\bar{E} = \bigcup_{i=1}^{m} \{i\} \times A_i$. The following are equivalent.
  \begin{enumerate}
    \item $E \in \mathcal E_{1,b}^{m,n}$.
    \item $A_1, \hdots, A_m$ have the $(n-b)$-dimensional saturation property.
  \end{enumerate}
\end{proposition}

\begin{proof}
  Apply Proposition~\ref{prop:mr-tensor} (a) and (c) with for a generic $(n-b)\times n$ matrix $W$.
\end{proof}

\begin{proof}[Proof of Theorem~\ref{thm:mds-tensor-equiv}]
  Consider $E \subseteq [m] \times [n]$. Let $A_1, \hdots, A_m \subseteq [n]$ be such that $\bar{E} = \bigcup_{i=1}^{m} \{i\} \times A_i$. Let $G$ be a generator matrix for $C$.

  By Proposition~\ref{prop:mr-tensor}, we have that $E$ is correctable for $C' \otimes C$ if and only if $A_1, \hdots, A_m$ is $G$-saturated. By Proposition~\ref{prop:sat-equiv}, we have that $C' \otimes C$ corrects every pattern in $\mathcal E_{1,b}^{m,n}$ if and only if every $A_1, \hdots, A_m$ with the $n-b$-dimensional saturation property is $G$-saturated.
\end{proof}

The following lemma is useful in proving that relaxed MR tensor codes exist over small fields.

\begin{lemma}\label{lem:mr-pad}
Let $E \in \mathcal E_{a,b}^{m,n}$. Pick $A \subseteq [m]$ and $B \subseteq [n]$. Let $E' = E \cup A \times [n] \cup [m] \times B$. Then, $E' \in \mathcal E_{a',b'}^{m,n}$, where $a' = \min(a+|A|, m)$ and $b' = \min(b+|B|, n)$.
\end{lemma}

\begin{proof}
By a suitable inductive argument, we may without loss of generality assume that $|A| = 1$ and $|B| = 0$. Also assume that $A = \{m\}$, so we have that $\bar{E'} = \bar{E} \cap [m-1] \times [n]$. If $a' = m$, this result follows from the fact that every pattern is correctable. Thus, assume $a' = a+|A|$.

Let $U$ be a generic $(m-a) \times m$ matrix and $V$ a generic $(n-b) \times n$ matrix.  By Proposition~\ref{prop:mr-tensor}, we have that $\sum_{(i,j) \in \bar{E}} U_i \otimes V_j = \F^{m-a} \otimes \F^{n-b}$ Let $\Pi : \F^{m-a} \to \F^{m-a-1}$ be a surjective linear map such that $\Pi(U_m) = 0$. Thus,
\[
\sum_{(i,j) \in \bar{E'}} \langle \Pi(U_i) \otimes V_j\rangle = \sum_{(i,j) \in \bar{E}} \langle \Pi(U_i) \otimes V_j\rangle = \F^{m-a-1} \otimes \F^{n-b}
\]
Thus, by Proposition~\ref{prop:mr-tensor}, $E'$ is correctable for $C_1 \otimes C_2$ for some $(m,m-a-1)$-code $C_1$ and $(n,n-b)$-code $C_2$. Thus, $E' \in \mathcal E_{a',b'}^{m,n}$, as desired.
\end{proof}

As a corollary, we get the following fact about the saturation property.

\begin{corollary}
Let $A_1, \hdots, A_m \subseteq [n]$ have the $k+d$-dimensional saturation property. Then, for any $B \subseteq [n]$ of size at most $d$, we have that $A_1 \setminus B, \hdots, A_m \setminus B$ have the $k$-dimensional saturation property.
\end{corollary}
\begin{proof}
  Let $E = \bigcup_{i=1}^m \{i\} \times \oA_i$ and $E' = \bigcup_{i=1}^m \{i\} \times (\oA_i \cup B)$. By Proposition~\ref{prop:sat-equiv}, it suffices to prove that if $E \in \mathcal E_{1,n-k-d}^{m,n}$ then $E' \in \mathcal E_{1,n-k}^{m,n}$. This follows form Lemma~\ref{lem:mr-pad}.
\end{proof}

\section{Constructions from Randomly Punctured Algebraic Codes}
\label{sec:agcode}
In this section, we will show that by randomly puncturing algebraic codes (such as Reed-Solomon codes or AG codes or even random linear codes), we can construct both upper and lower relaxed higher order MDS codes over small fields. We first recall some standard facts from algebraic geometry similar to \cite{bdg2023a}.

\subsection{Algebraic Preliminaries}

Let $\F$ be an algebraically closed field. An ideal $I\subset \F[x_1,\dots,x_k]$ is a subset of polynomials closed under addition and multiplication with arbitrary polynomials in $\F[x_1,\dots,x_k]$. The set of common zeros of a collection of polynomials in $\F[x_1,x_2,\dots,x_k]$ is called a \emph{variety}. The variety $V(I)$ denotes the common zeros of polynomials in an ideal $I$. The set of all polynomials vanishing on a subset $V\subset \F^k$ is an ideal in $\F[x_1,x_2,\dots,x_k]$ and is denoted by $I(V)$. 
An ideal $I\subseteq \F[x_1,\hdots,x_k]$ is said to be radical if $f^r\in I$ for any $r$ implies $f\in I$. 
By Hilbert's Nullstellensatz, there is a one-to-one correspondence between radical ideals and varieties, i.e., for a variety $X$, $I(X)$ is a radical ideal and $V(I(X))=X$. 
An irreducible variety $X$ is a variety for which $I(X)$ is a prime ideal.

A Zariski open set over a variety $X$ is a set of the form $U_f = \{x | f(x) \ne 0, x \in X\}$, where $f\in \F[x_1,\hdots,x_k]$. The Zariski topology over a variety $X$ is generated by the Zariski open sets over $X$. Note that $X$ is irreducible if and only if no subset of $X$ can be partitioned into two non-empty Zariski open sets. If $U_f$ is Zariski dense in $X$, then we say that $f$ is generically non-vanishing. For irreducible $X$, it is equivalent check that $f$ is nonzero for one point of $X$ or that $f \not\in I(X)$. We now state some simple facts about varieties.

\begin{lemma}[Product of varieties (See 10.4.H in \cite{raviVakilFOAG})]\label{lem-prodIrred}
If $X_1\subseteq \F^{k_1}$ and $X_2\subseteq \F^{k_2}$ are two varieties then $X_1\times X_2\subseteq \F^{k_1+k_2}$ is also a variety. Furthermore, if $X_1$ and $X_2$ are irreducible then so is $X_1\times X_2$.
\end{lemma}

A non-empty variety $X$ is said to have dimension $r$ if intersecting $X$ with $r$ generic hyperplanes gives a finite set of points. If we count the number of points projectively, that is take $\F_q^n$ as a subset of $\P \F_q^{n}$ and consider hyperplanes in projective space, then this number of points is generically constant and is said to be the degree of $X$. Bezout's theorem for curves states two irreducible curves of degree $d_1$ and $d_2$ intersect in at most $d_1d_2$ points. The following is a generalization to varieties. 

\begin{lemma}[Bezout's theorem (see 18.6.J in \cite{raviVakilFOAG})]\label{lem:bezout}
If $X$ is an irreducible curve of degree $d_1$ cut out by $I$ and $f$ is a polynomial of degree $d_2$ not vanishing on it then the $f$ has at most $d_1d_2$ zeros on $X$.

In general, if $X$ is an irreducible variety of degree $d_1$ and $f$ is a polynomial of degree $d_2$ not vanishing on it then $\{x|f(x)=0\}\cap X$ has degree at most $d_1d_2$.
\end{lemma}

We also need the following lemma. It can proved by noting that any irreducible variety has a local parametrization using power series in some (Zariski) open neighborhood of a point in it. So to check if a polynomial is generically non-vanishing on an irreducible variety $X$ it is enough to check if it is non-vanishing on the power series parametrization. A proof can be found in the appendix of \cite{bdg2023a}.

\begin{lemma}[See Lemma 6.2 in \cite{bdg2023a}]\label{lem-powerS}
For an algebraically closed $\F$ and an irreducible variety $X\subseteq \F^k$, there exists a $d$ ($d$ will be the dimension of $X$) such that we can find functions $f_1,\hdots,f_k$ (depending on $X$) in the formal power series ring $\F[[z_1,\hdots,z_d]]$ such that for any polynomial $g\in \F[x_1,\hdots,x_k]$ over $\F^k$, it is generically non-vanishing over $X$ if and only if $g(f_1,\hdots,f_k)\in \F[[z_1,\hdots,z_d]]$ is non-zero.

We also have that for $X^n\subseteq \F^{kn}$ any polynomial $g\in \F[x_{1,1},\hdots,x_{1,k},\hdots,x_{n,1},\hdots,x_{n,k}]$ is generically non-vanishing over $X^n$ if and only if 
$$g(f_1(z_{1,1},\hdots,z_{1,d}),\hdots,f_k(z_{1,1},\hdots,z_{1,d}),\hdots,f_1(z_{n,1},\hdots,z_{n,d}),\hdots,f_k(z_{n,1},\hdots,z_{n,d}))$$ is non-zero in $\F[[z_{1,1},\hdots,z_{1,d},\hdots,z_{n,1},\hdots,z_{n,d}]]$. 
\end{lemma}

 We need the following consequence of the above lemma.

\begin{lemma}\label{lem:partInit}
Let $f$ be a degree $d$ polynomial over $n_1+n_2$ variables such that $f$ is generically non-vanishing over $X_1\times X_2$ where $X_1$ and $X_2$ are defined over $n_1$ and $n_2$ variables respectively. 

Then there exists a degree at most $d$ polynomial $g$ over $n_1$ variables such that $g$ is generically non-vanishing over $X_1$ and if $g(x_1)$ is non-zero then setting the $n_1$ variables in $f$ to $x_1$ gives a polynomial which is generically non-vanishing on $X_2$. 
\end{lemma}
\begin{proof}
Using Lemma~\ref{lem-powerS} there exists power series $h_1,\hdots,h_{n_2}$ such that a polynomial is generically non-vanishing over $X_2$ if on substitution it remains non-vanishing. 

Substituting $h_1,\hdots,h_{n_2}$ for $X_2$ in $f$ has to give a non-zero polynomial and one of its coefficients has to be generically non-vanishing over $X_1$ by assumption. Picking any such coefficient gives us $g$. The degree bound follows from the construction.
\end{proof}

\subsection{Generalized GM-MDS Theorem for Polynomial Codes}
We will be crucially using the generalized GM-MDS theorem for polynomial codes due to Brakensiek-Dhar-Gopi (\cite{bdg2023a}). We first give the relevant definitions from their work.

\begin{definition}[$\MDS(\ell)$ property for varieties, \cite{bdg2023a}]\label{def-MDSlvar}

For an algebraic closed $\F$, we say a variety $X\subseteq \F^k$ satisfies the $\MDS(\ell)$ property if and only if for every $k\times n$ matrix $V$ with columns of $V$ initialized from $X$, $V$ is generically a $[n,k]$-$\MDS(\ell)$ code.

Concretely this means for every $A_1,\hdots,A_\ell\subseteq [n]$ with $|A_i| \le k$ and $|A_1|+\hdots+|A_\ell|=(\ell-1)k$ satisfies  $W_{A_1}\cap\hdots\cap W_{A_\ell}=0$ (equivalently $A_1,\hdots,A_\ell$ satisfies the $k$-dimensional null intersection property) if and only if $det(\cG_{A_1,\hdots,A_\ell}(V))$ is generically non-vanishing on $X^{(\ell-1)k}$, where $W$ is a $k\times n$ generic matrix (that is a matrix whose every entry is an independent formal variable).

\end{definition}
Since $\MDS(2)$ is equivalent to the usual MDS condition, we call an $\MDS(2)$ variety just an $\MDS$ variety. We will use the following simple characterization of $\MDS$ varieties.

\begin{proposition}[$\MDS$ varieties are not contained in hyper-planes passing through the origin (Theorem 6.3 in \cite{bdg2023a})]\label{thm:MDS2VarChar}
A variety $X\subseteq \F^k$ is $\MDS$ if and only if each of its irreducible components is not contained in any hyper-plane passing through the origin.
In particular, an irreducible variety is $\MDS$ if and only if it is not contained in any hyperplane passing through the origin.
\end{proposition}

\begin{definition}[$\LDMDS(\le \ell)$ property for varieties, \cite{bdg2023a}]\label{def-LDMDSlvar}
For an algebraic closed $\F$, we say a variety $X\subseteq \F^k$ satisfies the $\LDMDS(\le \ell)$ property if and only if for ever $k\times n$ matrix $V$ with columns of $V$ initialized from $X$, the dual of $V$ is generically a $[n,n-k]$-$\MDS(\ell+1)$ code.

Concretely this means for any $A_1,\hdots,A_{\ell+1}\subseteq [n]$ with $|A_i| \le n-k$ and $|A_1|+\hdots+|A_{\ell+1}|=\ell(n-k)$, satisfies $W_{A_1}\cap\hdots\cap W_{A_{\ell+1}}=0$ if and only if $\det(\cG_{A_1,\hdots,A_\ell}(W^\perp))$
is generically non-vanishing on $X^{\ell(n-k)}$ where $W$ is a $(n-k) \times n$ generic matrix (that is a matrix whose every entry is an independent formal variable).
\end{definition}
The original GM-MDS theorem first conjectured by \cite{dau2014gmmds} and proved independently by \cite{lovett2018gmmds,yildiz2019gmmds} is equivalent to saying that generic Reed-Solomon codes are higher order MDS for all $\ell$, which was shown in \cite{bgm2022}. We can now restate the original GM-MDS theorem as follows.
\begin{theorem}[GM-MDS \cite{dau2014gmmds,lovett2018gmmds,yildiz2019gmmds,bgm2022}]\label{thm:gm-mds}
  The curve $X=\{(1,t,t^2,\dots,t^{k-1}): t\in \F\}$ is $\MDS(\ell)$ and $\LDMDS(\le\ell)$ for all $\ell\ge 1$.
\end{theorem}

\cite{bdg2023a} generalized the GM-MDS theorem for any irreducible $\MDS$ variety $X$, i.e., as long as $X$ is not contained in a hyperplane through the origin.

\begin{theorem}[Generalized GM-MDS \cite{bdg2023a} (Theorems 6.4 and 6.5 in \cite{bdg2023a})]\label{thm-mdsIrred}
For an algebraic closed $\F$, if an irreducible $X$ is $\MDS$ (i.e., $X$ is not contained in a hyperplane through the origin) then it is $\MDS(\ell)$ for all $\ell\ge 1$ and also $\LDMDS(\le \ell)$ for all $\ell \ge 1$.
\end{theorem}

Unlike \cite{guo2023randomly} and \cite{alrabiah2023randomly} which used Theorem~\ref{thm:gm-mds} to drive their results, our results will be based on Theorem~\ref{thm-mdsIrred}.

\subsection{Adapting \cite{guo2023randomly} and \cite{alrabiah2023randomly}}

Extending the argument of \cite{guo2023randomly} and \cite{alrabiah2023randomly}, we want to prove that random puncturing of AG codes give us codes which satisfy relaxed versions of the higher order MDS property (as seen earlier the lower relaxation gives us list decoding and upper relaxation imply results for MR tensor codes).
The main theorems we want to prove are the following.

\begin{theorem}[Lower Relaxation]\label{thm:probBoundSlackLower}
Let $X\subseteq \overline{\F}_q^{k}$ be an irreducible $\MDS$ variety and $S$ be a set of points on $X \cap \F_q^k$. Let $G$ be $k\times n$ matrix with columns initialized uniformly at random from $S$. Then $G$ is $\rLDMDS_{r}(\le L)$ (equivalently $G^\perp$ is a $\rMDS_r(L+1)$-code)
with probability at least
\begin{align*}
1-2^{(L+1)n}\binom{n}{r/2}2^{r(L+1)/2}\left(\frac{P_q(X,L)}{|S|}\right)^{r/2},&\quad \text{\em (when repeated columns are allowed in $G$)}\\
1-2^{(L+1)n}\binom{n}{r/2}2^{r(L+1)/2}\left(\frac{P_q(X,L)}{|S|-n}\right)^{r/2},&\quad \text{\em (when repeated columns are {\bf not} allowed in $G$)}
\end{align*} 

where $P_q(X,L)$ is the maximum number of zeros a degree $L$ polynomial generically non-vanishing on $X$ can have over $X\cap \F_q^k$.
\end{theorem}

\begin{theorem}[Upper Relaxation]\label{thm:probBoundSlackUpper}
Let $X\subseteq \overline{\F}_q^{k}$ be an irreducible $\MDS$ variety and $S$ be a set of points on $X\cap \F_q^k$. Let $G$ be $k\times n$ matrix with columns initialized uniformly at random from $S$ and $C$ be the code spanned by it. 
 $C$ is $\rMDS^r(L)$ with probability at least 
 \begin{align*}
1-C_{n,k,r,L}\binom{n}{r/2}2^{rL/2}\left(\frac{P_q(X,L-1)}{|S|}\right)^{r/2},&\quad \text{\em (when repeated columns are allowed in $G$)}\\
1-C_{n,k,r,L}\binom{n}{r/2}2^{rL/2}\left(\frac{P_q(X,L-1)}{|S|-n}\right)^{r/2},&\quad \text{\em (when repeated columns are {\bf not} allowed in $G$)}
\end{align*} 

where $P_q(X,\ell)$ is the maximum number of zeros a degree $\ell$ polynomial generically non-vanishing on $X$ can have over $X\cap \F_q^k$ and 
$$C_{n,k,r,L}=\min\left[2^{Ln},\binom{(n-k-r)L(L-1)}{\le 2(n-k-r)(L-1)}\binom{n}{\le (n-k-r)(L-1)}+\binom{n}{k+r}\right].$$
Note that $C$ being $\MDS^r(L)$ is equivalent to $C'\otimes C$ being a $(1,n-k-r)$-$\rMR(L,n,1,n-k)$ code where $C'$ is a fixed $(L,L-1)$ parity check code.
\end{theorem}

Note when puncturing a code (that is $G$ does not have repeated columns), for the theorems to make sense, we need $|X\cap \F_q^k|\gg n$, i.e., there should sufficiently many $\F_q$ rational points on the variety $X$.
For our applications to Reed-Solomon codes (which recovers the result of \cite{alrabiah2023randomly}) and AG codes, $X$ will be a curve (that is a dimension $1$ irreducible variety). In this case, we have $P_q(X,\ell)\le \textsf{deg}(X)\ell$ by Bezout's Theorem (Lemma \ref{lem:bezout}). Another application would be to take $X=\overline{\F_q}^k$ which would give $P_q(\overline{\F}_q^k,\ell)\le \ell q^{k-1}$ by the Schwartz-Zippel lemma. In general one could use Bezout's Theorem and explicit Lang-Weil bounds (like the ones here \cite{CAFURE2006155,slavov_2023}) to bound this for general varieties.

We will now briefly describe the strategy of \cite{guo2023randomly} and \cite{alrabiah2023randomly} for Reed-Solomon codes.\footnote{They use a different family of matrices called \emph{reduced intersection matrices} for their result which comes from the list decoding interpretation. Hence, their result will only give us a lower relaxation for Reed-Solomon codes. We use our framework of matrices obtained by the two relaxations.}
Given a collection of sets $A_1,\hdots,A_\ell \subseteq [n]$, to prove Theorems~\ref{thm:probBoundSlackLower} and \ref{thm:probBoundSlackUpper} we need to show  that certain matrices depending on this set collection have some fixed rank when we randomly initialize $G$. By Theorem~\ref{thm-mdsIrred} these matrices have the same rank generically over $\MDS$ varieties as their rank as completely generic matrices.

As the argument for both theorems is identical, let us arbitrarily choose one setting. Say we want to construct $(n,k)$-code which is $\rMDS_r(\ell)$ codes by choosing a generator $G$ with columns sampled uniformly from some subset $S$ of a variety $X$. We can convert the $\rMDS_r(\ell)$ property to a collection of matrices having a prescribed rank.  Let $M(x_1,\hdots,x_n)$, where $x_i$ is a point in $\F^k$, be one such matrix whose rank we want to preserve. Say generically $M$ has rank $t$ ($t$ will have different values depending on the statement we are proving). We want to show that with high probability if we initialize $G$ with columns chosen randomly from some subset $S\subset X$, and substitute the $j^{th}$ column of $G$ for $x_j$, then $M(G)$ has rank $t$ with high probability. The following definition of a faulty index is quite useful for this.

\begin{definition}[Faulty indices]
  Let $X\subset \F^k$ be an irreducible variety.
  Let $P(x_1,\hdots,x_n)$ be a polynomial on $k\cdot n$ variables which is generically non-vanishing over $X^n$, where we substitute the $j$th copy of $X$ in $x_j$. Let $G=[\alpha_1,\hdots,\alpha_n]$ be a $k\times n$ matrix formed by choosing each column from $X$, i.e., $(\alpha_1,\alpha_2,\dots,\alpha_n)\in X^n$. We say $i$ is a faulty index for the pair
   $(P,G)$ if 
   $$P(\alpha_1,\hdots,\alpha_{i-1},x_i,\hdots,x_n)$$ 
  is generically non-vanishing over $X^{n-i+1}$, but
  $$P(\alpha_1,\hdots,\alpha_{i},x_{i+1},\hdots,x_n)$$
  vanishes on $X^{n-i}$. Usually $G$ will be clear from context and we just say $i$ is the faulty index of $P$.
  \end{definition}
  
In \cite{guo2023randomly} we  pick a $t\times t$ minor $M_1$ of $M$ which is a proof of the generic rank of $M$. Suppose $M_1(G)=0$, then we identify a faulty index $j_1$ for $(M_1,G)$. Because of Corollary~\ref{cor:slackLDMDSrank} (we will need something else for the other theorem) we have that deleting columns corresponding to $r$ indices in $[n]$ will not change the generic rank of these matrices. In particular, on deleting $x_{j_1}$ in $M$ we still have the correct generic rank $t$. Now we can find a new $t\times t$ minor $M_2$ and if $M_2(G)=0$, we find a new faulty index $j_2$. We can repeat this process at least $r$ times whenever $G$ causes a failure. The probability of this process given a fixed $d$ sequence of faulty indices can be bounded by $\left(\frac{(k-1)\ell}{|S|-n}\right)^{r}$ (when repetitions are not allowed) with a simple and clever conditioning argument. The problem is that one has to union bound over all possible $d$ tuple of faulty indices which are $n^d$ in number and $d$ is linear in $n$ in applications. This gives them a quadratic field size bound for punctured Reed-Solomon codes achieving list decoding capacity.

In \cite{alrabiah2023randomly} this argument is improved. The idea is to use the symmetries in $M$ to carefully choose the next minor when a failure happens. The symmetries use the following notion.

\begin{definition}[Type of a index]
An index $i\in [n]$ has type $\tau\subseteq [\ell]$ for the family of sets $A_1,\hdots,A_\ell\subseteq [n]$ where $\tau=\{j: i\in A_j\}$.
\end{definition}

In \cite{alrabiah2023randomly}, they note that if $M$ has some non-vanishing minor then permuting variables of the same type will give another non-vanishing minor of $M$. To be precise if $i$ and $j$ have the same type then $M_1(x_i,x_j)$ (assuming other variables as fixed) is a non vanishing minor if and only if $M_1(x_j,x_i)$. Next they set aside $r/2$ indices of each type and remove them from $M$ to get $M'$. They follow the \cite{guo2023randomly} procedure on $M'$ and after they find a faulty index for a minor in $M'$ they replace that with an index of the same type to get a new non-vanishing minor. We repeat by again trying to find a faulty index. This is done until no replacement is possible (replacement will be possible for at least $r/2^{\ell+1}$ steps). At that point they permanently get rid of all faulty indices from $M$ and repeat the process. This process will at least give $r/2$ faulty indices but the sequence of indices are guaranteed to have more structure which leads to a smaller union bound.
 
  We will instantiate this argument for $\MDS$ varieties using more implicit arguments and Bezout's Theorem (Lemma~\ref{lem:bezout}). To get the best bounds we will primarily adapt the framework in \cite{alrabiah2023randomly}.

We will now describe the \cite{alrabiah2023randomly} method in detail. Given a matrix of size $a\times b$ we put in a lexicographical order on its $t\times t$ minors by ordering them first by the list of columns chosen and then by rows (the choice of order does not matter, we just need any order which would stay consistent on deletion of columns).

We will abstract the proof into a framework so that it works for both Theorem~\ref{thm:probBoundSlackLower} and \ref{thm:probBoundSlackUpper}. In the following discussion, for Theorem~\ref{thm:probBoundSlackLower}, $W$ will be the null intersection property and $M$ will be $\cH$ (see \eqref{eq:Hop}) and for Theorem~\ref{thm:probBoundSlackUpper}, $W$ will be the saturation property and $M$ will be $\cG$ (see \eqref{eq:Gop}).

\begin{framework}\label{frmwrk:relax}
 We work with a family of matrices $M_{[A_i]_{i=1}^\ell}$ (later initialized to $\cG$ and $\cH$) dependent on a $k\times n$ matrix $G$ whose columns are initialized by points in $X$ where $X$ is an $\MDS$ irreducible variety in $\overline{\F}_q^k$. The following properties are assumed to hold:
\begin{enumerate}
\item There exists a property $W(n,k,r,\ell)$ on the sets $A_1,\hdots,A_\ell$ which if satisfied implies $M_{[A_i]_{i=1}^\ell}$ has generically rank at least $w(n,k,r,\ell)$ over $X^n$.
 \begin{itemize}
 \item For Theorem~\ref{thm:probBoundSlackLower} $W(n,k,r,\ell)$ is the $n-k-r$ null intersection property, $M=\cH$, and $w(n,k,r,\ell)=n+(\ell-1)k$. This property follows from the fact that $X$ is a $\LDMDS(\ell)$ variety by Theorem~\ref{thm-mdsIrred} and applying Corollary~\ref{cor:slackLDMDSrank} with the fact that a $\MDS(\ell)$ code is also lower relaxed.
 \item  For Theorem~\ref{thm:probBoundSlackUpper} $W(n,k,r,\ell)$ is the $k+r$ saturation property and not sharing an element, $M=\cG$, and $w(n,k,r,\ell)=\ell k$. This property follows from the fact that $X$ is a $\MDS(\ell)$ variety by Theorem~\ref{thm-mdsIrred}, hence is also generically upper relaxed, and noting that in the definition of the upper relaxation it suffices to look at saturated sets which do not share an element as sharing an element does not affect the rank of $\cG$ (as the kernel of $\cG$ exactly tells us the dimension of  $G_{A_1}\cap \hdots\cap G_{A_\ell}$, every shared element increases the kernel dimension by $1$).
 \end{itemize}
\item The method $\getMinor(M_{[A_i]_{i=1}^\ell},w(n,k,r,\ell))$ outputs the lexicographically earliest $w(n,k,r,\ell)\times w(n,k,r,\ell)$ minor of $M_{[A_i]_{i=1}^\ell}$ which generically does not vanish on $X^n$.
\begin{itemize}
 \item For both Theorem~\ref{thm:probBoundSlackLower} and Theorem~\ref{thm:probBoundSlackUpper} as the previous property gives us a rank bound this operation is well defined.
 \end{itemize}
\item Each column of $M_{[A_i]_{i=1}^\ell}$ depends on at most $1$ column of $G$.
 \begin{itemize}
 \item For both Theorem~\ref{thm:probBoundSlackLower} and Theorem~\ref{thm:probBoundSlackUpper} this can be easily checked by inspection of $\cH$ and $\cG$ respectively.
 \end{itemize}
\item For a collection $A_1,\hdots,A_\ell$ satisfying $W(n,k,r,\ell)$, the number of sets in which an index $\beta$ can appear is at most $\max(1,\ell-1)$ and the number of columns of $M_{[A_i]_{i=1}^\ell}$ depending on $\beta$ is exactly equal to the number of $A_i$s in which $\beta$ appears.
\begin{itemize}
 \item For Theorem~\ref{thm:probBoundSlackLower} if a collection of set satisfies the null intersection property then they can not share an element in common. The other property follows by the definition of $\cH$.
 \item For Theorem~\ref{thm:probBoundSlackUpper} $w(n,k,r,\ell)$ includes the property that $A_1,\hdots,A_\ell$ do not share an element. The other property follows by the definition of $\cH$.
 \end{itemize}
\item For a set $S\subseteq [n]$, we let $M_{[A_i]_{i=1}^\ell}^S$ refer to the sub-matrix of $M_{[A_i]_{i=1}^\ell}$ with only the columns corresponding to $S$ and the columns independent of any column in $G$ included.  $M_{[A_i]_{i=1}^\ell}^S$ will equal $M_{[A_i\cap S]_{i=1}^\ell}$ which depends on $G|_S$.  
\begin{itemize}
 \item For Theorem~\ref{thm:probBoundSlackLower} if we look at $\cH^{\oB}_{A_1,\hdots,A_\ell}(G)$ only depends on $G|_{\oB}$ and if we remove zero rows in the identity part depending on $B$ we get $ \cH_{A_1\cap \oB,\hdots,A_\ell \cap \oB}(G|_{\oB})$.
 \item For Theorem~\ref{thm:probBoundSlackUpper} this is a simple check.
 \end{itemize}
\item For any $B\subseteq [n],|B|\le r$, $\getMinor(M_{[A_i]_{i=1}^\ell}^{\overline{B}},w(n-|B|,k,r-|B|))$ will also be a minor of $M_{[A_i]_{i=1}^\ell}$ (note this just means that  determinant of the sub-matrices match, the sub-matrices could be of different sizes)
\begin{itemize}
 \item For Theorem~\ref{thm:probBoundSlackLower} if we look at a minor of $ \cH_{A_1\cap \oB,\hdots,A_\ell \cap \oB}(G|_{\oB})$ if we add an identity part corresponding to points in $B$ we get a minor in $\cH_{A_1,\hdots,A_\ell}(G)$ (by Proposition~\ref{prop:LDMDSrank-puncture}).
 \item For Theorem~\ref{thm:probBoundSlackUpper} this is a simple check.
 \end{itemize}
\item  For any $B\subseteq [n],|B|\le r$, $A_1\cap \oB, \hdots,A_\ell\cap \oB$ satisfies the property $W(n-|B|,k,r-|B|))$.
\begin{itemize}
 \item For Theorem~\ref{thm:probBoundSlackLower} if a collection $A_1,\hdots,A_\ell$  satisfies the $n-k-r$ null intersection property it is easy to check that removing points does not change that.
 \item For Theorem~\ref{thm:probBoundSlackUpper} it is again easy to check that if $A_1,\hdots,A_\ell$ satisfies the $k+r$ saturation property then removing an element changes the rank of $\cG$ by at most $\ell$ which is why even after removing at most $r$ points $B$, $A_1\cap \oB,\hdots,A_\ell\cap \oB$ satisfies the $k$ saturation property.
 \end{itemize}

\item Given any generically non-vanishing minor of $M_{[A_i]_{i=1}^\ell}$ if we swap two columns of $G$ corresponding to two indices of the same type then we will give us another generically non-vanishing minor.
 \begin{itemize}
 \item For both Theorem~\ref{thm:probBoundSlackLower} and Theorem~\ref{thm:probBoundSlackUpper} this can be easily checked by inspection of $\cH$ and $\cG$ respectively.
 \end{itemize}
\end{enumerate}
\end{framework}

\begin{theorem}\label{thm:mainSlackProb}
Let $M_{[A_i]_{i=1}^\ell}$ be a family of matrices, $X\subseteq\overline{\F}_q^k$ be an irreducible $\MDS$ variety $X$ with $W,w$ satisfying Framework~\ref{frmwrk:relax}. $G$ is randomly initialized with points in $S\subset \F_q^k\times X$.
Now, for a given $\ell>0$ and a given $A_1,\hdots,A_\ell$ satisfying property $W(n,k,\ell)$   $M_{[A_i]_{i=1}^\ell}$ has rank $w(n,k,r,\ell)$ with probability (over the choice of $G$) at least,
\begin{align*}
1-\binom{n}{r/2}2^{r\ell/2}\left(\frac{P_q(X,\ell-1)}{|S|}\right)^{r/2},& \text{\em (when repeated columns are allowed in $G$)}\\
1-\binom{n}{r/2}2^{r\ell/2}\left(\frac{P_q(X,\ell-1)}{|S|-n}\right)^{r/2},& \text{\em (when repeated columns are {\bf not} allowed in $G$)}
\end{align*} 

where $P_q(X,\ell)$ is the maximum number of zeros a degree $\ell$ generically non-vanishing polynomial can have over $X\cap \F_q^k$.
\end{theorem}

We first see that the above lemma encapsulates both Theorem~\ref{thm:probBoundSlackUpper} and \ref{thm:probBoundSlackLower}.

\begin{proof}[Proof of Theorem~\ref{thm:probBoundSlackUpper} given Theorem~\ref{thm:mainSlackProb}]
For Theorem~\ref{thm:probBoundSlackUpper} use $M=\cG$, $W(n,k,r,\ell)$ is having the $k+r$ dimensional saturation property and not sharing an element. $w(n,k,r,\ell)=\ell k$.

The result then follows from Theorem~\ref{thm:mainSlackProb} and Theorem~\ref{thm:mds-tensor-equiv} and union bounding over an appropriate set of families $A_1,\hdots,A_{\ell}$. A simple bound on the number of families would be $2^{Ln}$ but we can do better in some cases. 

We first ensure the code is $\MDS^r(2)$ so that by (d) in Proposition~\ref{prop:upper-mds} we have the code generated by $G$ has distance at least $n-k-r+1$.
\begin{claim}
Take a list $((A_1,B_1),\hdots,(A_{\binom{n}{k+d}},B_{\binom{n}{k+d}}))$ of pairs of subsets of $[n]$ such that $A_i\cap B_i=\phi,i\in [\binom{n}{k+d}]$ and $A_i\cup B_i$ gives us all possible $k+d$ sized subsets of $[n]$. If $\cG_{[A_i,B_i]}$ has rank $2k$ for a $k\times n$ code $G$ then $G$ is $\MDS^d(2)$.
\end{claim}  
\begin{proof}
We first note that for $G$ to be $\MDS^d(2)$ what we need is that any $k+d$ column of $G$ span $\F^k$ (by (d) in Proposition~\ref{prop:upper-mds}). Given any set $S$ of $k+d$ column we can find $A_j,B_j$ in our list such that their union is $S$. If $\cG_{[A_j,B_j]}$ is of rank $2k$ then that means the intersection of the span of $G|_{A_j}$ and $G|_{B_j}$ is of dimension $d$ which would imply the span of $G|_{A_j\cup B_j}=G|_S$ is of dimension $|A_j|+|B_j|-d=k$. As this can be done for any $k+d$ size subset of $G$ we are done.
\end{proof}
By the above claim we need to use $\binom{n}{k+r}$ patterns to be $\MDS^r(2)$ (we note that two disjoint subsets of $[n]$ such that their union is of size $k+r$ has the $k+r$ dimensional saturation property). By Lemma~\ref{lem:noOfMRpatterns} and Theorem~\ref{thm:mds-tensor-equiv} we see that checking an additional $|\mathcal{P}_{1,n-k-r}^{L,n}|$ many patterns suffices to check if $G$ is $\MDS^r(L)$. This gives us the desired value for $C_{n,k,r,L}$.
\end{proof}

\begin{proof}[Proof of Theorem~\ref{thm:probBoundSlackLower} given \ref{thm:mainSlackProb}]

For Theorem~\ref{thm:probBoundSlackLower}, $M=\cH$ , $W(n,k,r,\ell)$ is $n-k-r$ dimensional null intersection property and $w(n,k,r,\ell)=n+(\ell-1)k$. To get the $\rLDMDS_r(\le L)$ property we need to apply Theorem~\ref{thm:mainSlackProb} for $\ell \in \{2,\hdots,L+1\}$. 

The result then follows from Corollary~\ref{cor:slackLDMDSrank} and Theorem~\ref{thm:mainSlackProb} union bounding over the set families $A_1,\hdots,A_\ell\subseteq [n],\ell \in \{2,\hdots,L+1\}$ which satisfy the $n-k-r$ null intersection property. This number is at most $2^{(L+1)n}$ (because it suffices to consider families with no empty sets).
\end{proof}

\begin{proof}[Proof of Theorem~\ref{thm:mainSlackProb}]

For a fixed $\ell\in \{1,\hdots,\ell\}$ and $A_1,\hdots,A_\ell$ satisfying $W(n,k,r,\ell)$ we are going to adapt the algorithm in \cite{alrabiah2023randomly} to find faulty indices for minors in $M_{[A_i]_{i=1}^\ell}$ for a fixed initialization of $G$ with columns $\alpha_1,\hdots,\alpha_n\in X$.

\begin{algorithm}[h]
\caption{Algorithm to find faulty indices}\label{alg:FAULTY}
\begin{algorithmic}[1]
\State \textbf{Input:} Variety $X$, $\balpha\in X^n$, and $A_1,\hdots,A_\ell\subset[n]$
\State $S_\tau \gets \phi$ for all $2^\ell$ types $\tau$
\State Let $B,D,R$ be empty sequences 
\State $\textsf{REFRESH}\gets \textsf{YES}$
\While{$|B|\le r/2$}
\If{$\textsf{REFRESH}$ equals $\textsf{YES}$}
\State  {\em \%\% Comment: We call this the REFRESH if clause}
	\State Append $|B|+1$ to $R$
	\State $S_\tau\gets r/2^{\ell+1}$ largest indices in $\cup_{t=1}^\ell A_t\setminus B$ of type $\tau$
\State	$\textsf{REFRESH}\gets\textsf{NO}$
	\State $S\gets \cup_{\tau\subseteq [\ell]} S_\tau$
\State	$M_{curr}\gets \getMinor\left(M^{\overline{B\cup S}}_{[A_i]_{i=1}^\ell},w(n-|B|-|S|,k,r-|B|-|S|)\right)$ 
\EndIf
\If{$M_{curr}$ does not vanish over $\alpha_1,\hdots,\alpha_n$}
\State \textbf{Output:} SUCCESS
\EndIf
\State Let $\beta$ be the faulty index for $(M_{curr},(\alpha_1,\hdots,\alpha_n))$ and have type $\tau_\beta$
\State Append $\beta$ to $B$, Append $M_{curr}$ to $D$
 \If{$S_{\tau_\beta}$ is not empty}
    \State Let $\gamma$ be the smallest element in $S_{\tau_\beta}$
    \State $S_{\tau_\beta}\gets S_{\tau_\beta}\setminus \{\gamma\}$
    \State Substitute every occurrence of $\alpha_\beta$ coordinates in $M_{curr}$ to $\alpha_\gamma$ coordinates. 
  \Else
   \State $\textsf{REFRESH} \gets \textsf{YES}$
   \EndIf 
\EndWhile
\State Append $r/2+1$ to $R$
\State \textbf{Output:} $B=\beta_1,\hdots,\beta_{r/2}$, $D=M_1,\hdots,M_{r/2}$ and $R$
\end{algorithmic}
\end{algorithm}

We first want to show that if Algorithm~\ref{alg:FAULTY} outputs SUCCESS then $M$ has the desired rank.

\begin{claim}\label{cl:AlgSucc}
After initialization $M_{curr}$ in Algorithm~\ref{alg:FAULTY} is always a generically non-vanishing $w(n,k,r,\ell)\times w(n,k,r,\ell)$ minor of $M_{[A_i]_{i=1}^\ell}$ and if Algorithm~\ref{alg:FAULTY} outputs SUCCESS then $M_{[A_i]_{i=1}^\ell}(\alpha_1,\hdots,\alpha_n)$ has rank at least $w(n,k,r/2,\ell)$.
\end{claim}
\begin{proof}
Throughout the runtime of Algorithm~\ref{alg:FAULTY} it is easy to see that $|B|+|S|\le r$. This ensures the $\getMinor$ procedure in the algorithm is correctly called on line 12. (by property 6 of Framework~\ref{frmwrk:relax}) and gives us a generically non-vanishing minor of $M^{\overline{C}}_{[A_i]_{i=1}^\ell}$ for some $C\subseteq [n],|C|\le r$ as $M_{curr}$. Therefore by property 6 in Framework~\ref{frmwrk:relax} we see that $M_{curr}$ is also a $w(n,k,r,\ell)\times w(n,k,r,\ell)$ minor of $M_{[A_i]_{i=1}^\ell}$. Property 8 of Framework~\ref{frmwrk:relax} ensures that switching of variables of the same type in line 22 of Algorithm~\ref{alg:FAULTY} still gives us a generically non-vanishing minor. This shows that $M_{curr}$ is always a generically non-vanishing $w(n,k,r,\ell)\times w(n,k,r,\ell)$ minor of $M_{[A_i]_{i=1}^\ell}$.

If the algorithm outputs SUCCESS then at some point $M_{curr}$ is non-vanishing over $(\alpha_1,\hdots,\alpha_n)$. 
\end{proof}

We next need to show some properties of $B,D,R$ in Algorithm~\ref{alg:FAULTY}.

\begin{claim}\label{cl:BDRinvar}
If Algorithm~\ref{alg:FAULTY} outputs $B,D,R$ then the indices in $B$ are distinct and the matrices in $D$ are always minors of $M_{[A_i]_{i=1}^\ell}$ and the $j$th entry of $B$ is the faulty index for the $j$th entry of $D$.
\end{claim}
\begin{proof}
Whenever we add a faulty index to $B$ all columns corresponding to it are removed while considering $M_{[A_i]_{i=1}^\ell}$ in the future. This means entries in $B$ cannot repeat. 
The property about $D$ follows from the fact that $M_{curr}$ is always a minor from the previous claim and line 17 and 18 in Algorithm~\ref{alg:FAULTY}.
\end{proof}

Finally, we need a claim about some extra structure on $B,D,R$ which will let us reduce the number of cases in our union bound.

\begin{claim}\label{cl:propAl}
Algorithm~\ref{alg:FAULTY} satisfies the following properties,
\begin{enumerate}
\item If the algorithm outputs $B,D,R=(j_1,\hdots,j_t)$ then the $j_i$th to $j_{i+1}-1$th entries in $B$ are in increasing order for all $i$. We also have $j_{i+1}-j_i\ge r/2^{\ell+1}$ for $i<t-1$. As $|B|=r/2$ this means that $B$ contains $2^\ell$ increasing runs of indices.
\item If the algorithm outputs $B,D,R$ then $B$, $X$, and the set system $A_1,\hdots,A_\ell$ completely determines $D$ and $R$ (in particular we do not need to know the values $\alpha_1,\hdots,\alpha_n$).
\end{enumerate}
\end{claim}
\begin{proof}
We first prove 1. We want to show if $R=(j_1,\hdots,j_t$ then consecutive entries differ by at least $r/2^{\ell+1}$ (except possibly for the last two). This is because when we first add $j_i$ to $R$ the next $M_{curr}$ is chosen in the REFRESH if clause. At that point we remove columns corresponding to entries in $B$ and $S_{\tau}$ the at most $r/2^{\ell+1}$ largest entries in $[n]\setminus B$ of type $\tau$. This means indices in the new $M_{curr}$ must be outside of $B\cup \cup_{\tau\subseteq [\ell]} S_{\tau}$. For any such index if their type is $\tau'$ then $|S_{\tau'}|\ge r/2^{\ell+1}$ (as otherwise they would be included in $S_{\tau'}$). $R$ will not have a new addition until one such $S_\tau$ becomes empty which will take at least $r/2^{\ell+1}$ steps. 

To complete the proof of 1., we want to show that the sequence in $B$ is increasing between the $j_i$th and $j_{i+1}-1$th entry. Take any two consecutive entries of $\beta_1$ and $\beta_2$ in this range. Say $\beta_1$ is of type $\tau$. When $\beta_1$ is added to $B$ let $M_{old}$ be the minor added to $M$. The next $M_{curr}$ will be obtained by replacing the $\beta_1$ variables with variables corresponding to an index $\gamma$ of type $\tau$ with $\gamma>\beta_1$. We will call that $M_{new}$.  If $\beta_2$ is less than $\beta_1$ then that means  $M_{new}(\alpha_1,\hdots,\alpha_{\beta_2})$ is not generically non-vanishing in the remaining variables which include variables corresponding to $\gamma$ and $\beta_1$. If we switch these two we still have that $M_{old}(\alpha_1,\hdots,\alpha_{\beta_2})$ is not generically non-vanishing  which contradicts the fact that $\beta_1$ is the faulty index for $M_{old}$.

2. follows easily from noting that the columns in $M_{[A_i]_{=1}^\ell]}$ that that additions to $D$ are values in $M_{curr}$ which is determined by the previous value of $B$. Because $M_{curr}$ in line 12 of Algorithm~\ref{alg:FAULTY} is determined by $B$ and in line 22 the entry being replaced is determined by the new addition to $B$. A similar observation holds for $R$. 
\end{proof}

The next two claims with a union bound give us the probability that for a fixed $\ell\in \{1,\hdots,\ell\}$ and $A_1,\hdots,A_\ell$ $M_{[A_i]_{i=1}^\ell}$ does not have the right rank. This is because by Claim~\ref{cl:AlgSucc} we output $B,D,R$ if the right rank is not achieved.

\begin{claim}
The number of possible sequences $B,D,R$ that the algorithm can output is 
$$\binom{n}{r/2} 2^{\ell r/2}.$$ 
\end{claim}
\begin{proof}
By 2 in Claim~\ref{cl:propAl}, it suffices to just count the possible number of $B$.
By Claim~\ref{cl:BDRinvar}, $B$ contains $r/2$ distinct indices which can be chosen in $\binom{n}{r/2}$ ways. By 1 in Claim~\ref{cl:propAl} we see that $B$ contains $2^\ell$ increasing runs of indices. The $r/2$ numbers can be split over these runs in at most $2^{\ell r/2}$ ways.
\end{proof}

\begin{claim}
The probability over the choice of $\alpha_1,\hdots,\alpha_n$ that a given $B,D,R$ sequence is outputed by Algorithm~\ref{alg:FAULTY} is
\begin{align*}
\left(\frac{P_q(X,\ell-1)}{|S|}\right)^{r/2},& \quad\text{\em (When repeated columns are allowed in $G$)}\\
\left(\frac{P_q(X,\ell-1)}{|S|-n}\right)^{r/2},& \quad\text{\em (When repeated columns are {\bf not} allowed in $G$).}
\end{align*} 
\end{claim}
This is proven using a simple conditioning argument which appears in \cite{guo2023randomly,alrabiah2023randomly}. They prove the statement for Reed-Solomon codes using Schwartz-Zippel and \cite{alrabiah2023randomly} gives a different argument for random linear codes.  We adapt it to the setting of varieties.
\begin{proof}
Let us fix a particular output $B=(\beta_1,\hdots,\beta_{r/2}),D=(M_1,\hdots,M_{r/2}),R$. Now let us order the indices in increasing order $j_1,\hdots,j_{r/2}$ and let the corresponding minors be $M'_{1},\hdots,M'_{r/2}$. By Claim~\ref{cl:BDRinvar}, $j_t$ is the faulty index for $(M'_t,\balpha=(\alpha_1,\hdots,\alpha_n))$ for all $t$. Let $E_t$ be the event (depending on $\balpha$) that $j_t$ is the faulty index for $(M'_t,\balpha)$. 

We also let $C_t$ be the event that $M'_t$ is generically non-vanishing in the remaining variables after plugging in $(\alpha_1,\hdots,\alpha_{j_t-1})$ and $F_t$ be the even that $M'_t$ is not generically non-vanishing in the remaining variables after plugging in $(\alpha_1,\hdots,\alpha_{j_t})$. It is clear that $E_t=C_t\cap F_t$. For the output to be $(B,D,R)$ the events $E_1,\hdots,E_{r/2}$ must have taken place. We simply have to bound
$$\textsf{Prob}[E_1\cap\hdots\cap E_{r/2}]=\prod\limits_{i=1}^{r/2} \textsf{Prob}[E_i|E_1,\hdots,E_{i-1}].$$
   $\textsf{Prob}[E_i|E_1,\hdots,E_{i-1}]$ is now upper bounded by the probability that $F_i$ takes place conditioned on $E_1,\hdots,E_{i-1},C_i$ taking place. If the conditioned events happen with probability zero then $\textsf{Prob}[E_1\cap\hdots\cap E_{r/2}]$ has probability $0$ so we can assume that the conditioned can happen. Using 3 and 4 in Framework~\ref{frmwrk:relax} we see that $M_i$ is degree at most $\ell-1$ in the variables corresponding to $j_i$. Using Lemma~\ref{lem:partInit}, we see that there exists a generically non-vanishing on $X$ degree at most $\ell-1$ polynomial $g_i$ which vanishes on a point $\gamma\in X$ if $M_i(\alpha_1,\hdots,\alpha_{j_i-1},\gamma)$ is not generically non-vanishing on the remaining variables. By definition the number of points on $X$, $g_i$ can vanish on is at most $P_q(X,\ell-1)$. If $G$ is allowed to have repeated columns then the total number of points $\alpha_{j_i}$ can be chosen from is of size at least $|S|$ which shows that conditioned on $E_1,\hdots,E_{i-1},C_i$ (which at most fix $\alpha_t$ for $t<j_i$) $F_i$ occurs with probability at most
   $$\frac{P_q(X,\ell-1)}{|S|}.$$
Similarly when $G$ is not allowed to have repeated columns $G$ then $\alpha_{j_i}$ is chosen from at least $|S|-n$ elements which shows $F_i$ occurs with probability at most,
   \[\frac{P_q(X,\ell-1)}{|S|-n}.\qedhere\]
\end{proof}

We complete the proof by using the previous two claims and union bounding.
\end{proof}

We now recover the Reed-Solomon and Random linear codes result of \cite{alrabiah2023randomly} (we get a slight improvement because in our argument we allowed for constructing the generator matrix of the code with replacement).

\begin{corollary}[punctured RS codes are relaxed $\LDMDS$ over linear field sizes]
For any $\epsilon>0$, a $[q,k]$ Reed-Solomon code for $q\ge n+k 2^{10 L/\epsilon}$ on random puncturing to a $[n,k]$ code will be $\rLDMDS_{\epsilon n}(\le L)$ with probability at least $1-2^{-Ln}$.
\end{corollary}
\begin{proof}
We apply Theorem~\ref{thm:probBoundSlackLower} where $X$ is the irreducible $\MDS$ curve defined by the (closure) of the curve $(1,x,\hdots,x^{k-1})$ which is of degree $k-1$. $P_q(X,L)$ is at most $(k-1)L$ by Bezout's Theorem (Lemma~\ref{lem:bezout}) and $S$ are the $q$ points on the curve. If we set $r=\epsilon n$ then we see the failure probability for the punctured code being $\rLDMDS_{\epsilon n}(\le L)$ is at most

$$2^{(L+1)n}\binom{n}{\epsilon n /2}2^{\epsilon n (L+1)/2}\left(\frac{(k-1)L}{q-n}\right)^{\epsilon n /2}\le 2^{-Ln},$$

for $q\ge n+k2^{10L/\epsilon}$.
\end{proof}
With a similar calculation we get a slight improvement if we allow the generator to have repeated columns.

\begin{corollary}[Random RS codes are relaxed $\LDMDS$ over linear field sizes]
For any $\epsilon>0$, a $[n,k]$ Reed-Solomon code with generator $G$ for $q\ge k 2^{10 L/\epsilon}$ with each column of $G$ initialized randomly as $(1,\alpha,\hdots,\alpha^{k-1})$ with $\alpha\in\F_q$ will be $\rLDMDS_{\epsilon n}(\le L)$ with probability at least $1-2^{-Ln}$.
\end{corollary}

\begin{corollary}[random linear codes are relaxed $\LDMDS$ over constant field sizes]
For any $\epsilon>0$, rate $R>0$, $q\ge 2^{10 L/\epsilon}$, and $n$ large enough, a random linear code will be $\rLDMDS_{\epsilon n}(\le L)$ with probability at least $1-2^{-Ln}$.
\end{corollary}
\begin{proof}
Let $k=Rn$. The $\MDS$ variety $X$ we work with here is the affine plane $\overline{F}_q^{k}$. $P_q(X,\ell-1)$ is at most $(\ell-1)q^{k-1}$ using Schwartz-Zippel (can also follow from a repeated application of Bezout's Theorem). $S=\F_q^{k}$ and hence is of size $q^k$. Therefore the failure probability is 
at most
$$2^{(L+1)n}\binom{n}{\epsilon n /2}2^{\epsilon n (L+1)/2}\left(\frac{Lq^{k-1}}{q^k}\right)^{\epsilon n /2}\le 2^{-Ln},$$
for $q\ge 2^{10 L/\epsilon}$, $k=Rn$ and $n$ large enough.
\end{proof}
The same calculations as above but with Theorem~\ref{thm:probBoundSlackUpper} and Lemma~\ref{lem:noOfMRpatterns} prove the following two corollaries. This extends the results of \cite{guo2023randomly,alrabiah2023randomly} to the relaxed MR setting (via Theorem~\ref{thm:mds-tensor-equiv}).

\begin{corollary}[punctured RS codes give relaxed $\MR$ codes over linear field sizes (constant rate setting)]
For any $\epsilon>0$, a $[q,k]$ Reed-Solomon for $q\ge k 2^{10 L/\epsilon}+n$ on random puncturing to a $[n,k]$ code will be $\rMDS^{\epsilon n}(L)$ with probability at least $1-2^{-Ln}$.

If we allow repetitions in the generator matrix then a $[n,k]$ Reed-Solomon code with generator $G$ for a field size greater than $k 2^{10 L/\epsilon}$ with each column of $G$ initialized randomly as $(1,\alpha,\hdots,\alpha^{k-1})$ with $\alpha\in\F_q$ will be $\rMDS^{\epsilon n}(L)$ with probability at least $1-2^{-Ln}$.
\end{corollary}

\begin{corollary}[random linear codes give relaxed $\MR$ codes over constant field sizes (constant rate setting)]
For any $\epsilon>0$, a random linear code for $q\ge 2^{10 L/\epsilon}$ will be $\rMDS^{\epsilon n}(L)$ with probability at least $1-2^{-Ln}$.
\end{corollary}

Unlike the $\LDMDS$ setting, the relaxed MR setting also makes sense for constant/very low co-dimension codes (we skip stating the versions for $G$ with repeated columns as we do not get significant savings here).

\begin{corollary}[punctured RS codes give relaxed $\MR$ codes over polynomial field sizes (low co-dimension setting)]\label{cor:lowCodimMRRS}
For any $\epsilon>0$ and $k=n-b$, a $[q,k]$ Reed-Solomon for (the exact formula is at the end)
$$q\gg_{L,b} k n^{4(L-1)/\epsilon},$$
 on random puncturing to a $[n,k]$ code will be $\rMDS^{b\epsilon}(L)$ (leading to a $(1,(1-\epsilon)b)$-$\rMR(L,n,1,b)$ code) with probability at least $1-n^{-L}$.
\end{corollary}
\begin{proof}
For $L= 2$, RS codes are already $\MDS$ so we can take $L\ge 3$.
We apply Theorem~\ref{thm:probBoundSlackLower} where $X$ is the irreducible $\MDS$ curve defined by the (closure) of the curve $(1,x,\hdots,x^{k-1})$ which is of degree $k-1$. $P_q(X,L)$ is at most $(k-1)L$ by Bezout's Theorem (Lemma~\ref{lem:bezout}) and $S$ are the $q$ points on the curve. If we set $r=\epsilon b$ then we see the failure probability for the punctured code being $\rMDS^{\epsilon b}(L)$ is at most

\begin{align*}
\left(\binom{n}{\le b(1-\epsilon)(L-1)}\binom{L(L-1)b(1-\epsilon)}{\le b(1-\epsilon)(2L-2)}+\binom{n}{b(1-\epsilon)}\right) \binom{n}{b\epsilon /2}2^{\epsilon bL/2}\left(\frac{(k-1)L}{q-n}\right)^{\epsilon b/2}&\le \\
\left( \frac{en}{b(1-\epsilon)(L-1)}\right)^{b(1-\epsilon)(L-1)} (eL/2)^{2(L-1)b(1-\epsilon)} \left(\frac{en}{b(1-\epsilon)}\right)^{b(1-\epsilon)}2^{\epsilon bL/2} \left(\frac{(k-1)L}{q-n}\right)^{\epsilon b/2}&\le\\
\left( \frac{en}{b(1-\epsilon)}\right)^{b(1-\epsilon)L} (eL/2)^{2(L-1)b(1-\epsilon)} \left(\frac{(k-1)L}{q-n}\right)^{\epsilon b/2}&\le n^{-L}
\end{align*}

for $q\ge (k-1)Ln^{2L/(b\epsilon)}\left(\frac{e^3L^2n}{4b}\right)^{3L(1-\epsilon)/\epsilon}+n.$ 
\end{proof}
The striking feature about the above bound is how the exponent does not depend on $b$ but only on $\epsilon$ and $L$. By a similar calculation one can show that in this setting random linear codes can shave off a linear factor.
\begin{corollary}[random linear codes give relaxed $\MR$ codes over polynomial field sizes (low codimension setting)]
For any $\epsilon>0$ and $k=n-b$, a random linear code for 
$$q\gg_{L,b} n^{4(L-1)/\epsilon}$$
will be $\rMDS^{\epsilon b}(L)$ with probability at least $1-n^{-L}$.
\end{corollary}

\subsection{Application to AG Codes}

We first define evaluation codes. Algebraic geometric codes are instances of this. 

\begin{definition}[Evaluation codes]
Let  $F=\{f_1,\hdots,f_k\}$ be a set of rational functions in $\F_q(x_1,x_2,\hdots,x_r)$ and $P=(p_1,\hdots,p_n)\in (\F_q^r)^n$ such that the functions in $F$ do not have a pole over the points in $P$. (In other words substituting any point in $P$ in any function in $F$ gives us an element in $\F_q$.)

We define the $\Eval(F,P)$ evaluation code as the image of a linear map from the $\F_q^k$ to $\F_q^n$ where we map $(y_1,\hdots,y_k)$ to $(\sum_{i=1}^k y_if_i(p_1),\hdots,\sum_{i=1}^k y_if_i(p_n))$.
\end{definition}

Reed Solomon and Gabidulin (c.f., \cite{gabidulin2021rank}) codes are examples where $F$ are monomials in one variable.

A simple corollary of  Theorem~\ref{thm-mdsIrred} will imply that a large family of evaluation codes can attain list decoding capacity: if $p_1,\hdots,p_n$ are chosen generically over a $\MDS$ irreducible variety $X$ and the $f_1,\hdots,f_k$ are chosen to be linearly independent rational functions over $X$.

\begin{corollary}\label{cor:genptsMDS}
Let $X$ be an irreducible variety cut out by the prime ideal $I\subseteq \bar{\F}_q[x_1,\hdots,x_r]$. If $F=\{f_1,\hdots,f_k\}$ are linearly independent elements in the fraction field of $\bar{\F}_q[x_1,\hdots,x_r]/I$ then for all integers $n>0$, $\Eval(F,(p_1,\hdots,p_n))$ is $[n,k]-\MDS(\ell)$ and $[n,k]-\LDMDS(\le \ell)$ for a generic choice of $p_1,\hdots,p_n$ in $X^n$.

In other words the tuple of points $p_1,\hdots,p_n\in X$ for which $\Eval(F,(p_1,\hdots,p_n))$ is {\bf not} $\MDS(\ell)$ and $\LDMDS(\le \ell)$ is cut out by a non-vanishing ideal over $X^n$.
\end{corollary}
\begin{proof}
The generator matrix $\Eval(F,(p_1,\hdots,p_n))$ has columns $[f_1(p_i),\hdots,f_k(p_i)]$ for $i\in [n]$. 

Consider the map from $\bar{F}_q[y_1,\hdots,y_k]$ to the fraction field of $\bar{\F}_q[x_1,\hdots,x_r]/I$ which sends $y_i$ to $f_i(x_1,\hdots,x_r)$. Let the kernel of this map be $J$ (as it is a kernel of a ring homomorphism it is a prime ideal). By construction the set of points $\{(f_1(p),\hdots,f_k(p))|p\in X\}$ is the variety $X_J$ cut out by $J$. As $J$ is prime $X_J$ is irreducible.

As $f_1,\hdots,f_k$ are $\bar{\F}_q$ linearly independent in the fraction field of $\bar{\F}_q[x_1,\hdots,x_r]/I$, we have that $J$ does not contain any hyperplane passing through the origin.  Proposition~\ref{thm:MDS2VarChar} then implies that $X_J$ is $\MDS$. 

Now, Theorem~\ref{thm-mdsIrred} applied on $X_J$ gives us what we need.
\end{proof}

Algebraic-Geometric codes are evaluations codes where $F$ is chosen to be the spanning set of rational functions on an irreducible (projective) curve $X$ with some constraints on the multiplicity of poles and zeros on a fixed finite subset of $X$.

To state our theorems we will give a brief description. More details can be found in the following survey \cite{AGcodesurvey}. Here we will be interested in projective irreducible curves $C$. Projective curves are cut out by ideals of homogeneous polynomials in a finite number of variables and are irreducible if the ideal is prime. $\Proj\overline{\F}_q^{n}$ is the $n$ dimensional projective space and corresponds to the $0$ ideal in $n+1$ variables. If we remove any transverse hyperplane from an irreducible projective curve, we will get an irreducible affine curve so Theorem~\ref{thm-mdsIrred} still applies.

We assume $C$ is irreducible throughout for convenience. If $C$ is cut out by an ideal of homogeneous polynomials in $\F_q[x_1,\hdots,x_n]$ then rational functions on $C$ are ratios of homogeneous polynomials of the same degree such that in its reduced form the denominator is non-vanishing on $C$. A homogeneous polynomial $f$ can be assigned a multiplicity of vanishing on a point. $f$ is said to vanish on $x$ with multiplicity $\mult(f,x)=n$ if all Hasse derivatives of $f$ of order at most $n$ vanish on $x$. Each rational function $f/g$ where $f$ and $g$ are homogeneous polynomials can now be assigned multiplicity  $\mult(f,x)-\mult(g,x)$ at $x$.

A {\bf divisor} is a finite formal sum of points on $C$ with integer coefficients. A divisor $D$ is said to be $D\ge 0$ if it has only non-negative coefficients. The degree of a divisor is simply the sum of its coefficients. 

 Given a rational function $r$ we define $\divi(r)$ to be the divisor $\sum_{x\in C} \mult(r,x)x$. Such divisors are always degree $0$ (geometrically speaking a rational function has the same number of poles and zeros).

Let $D=\sum_{i=1}^t n_ip_i$ be a divisor where $p_1,\hdots,p_n$ are $\F_q$ points on $C$. $\mathcal{L}(D)$ is defined as a $\F_q$ vector space of rational functions $r$ which satisfy $\divi(r)+D\ge 0$. 

We are now ready to define AG codes (more details can be found in this survey~\cite{AGcodesurvey}).

\begin{definition}[Algebraic Geometry codes]
Let $C$ be a projective irreducible curve and $D=\sum_{i=1}^t n_ip_i$ be a divisor where $p_1,\hdots,p_n$ are $\F_q$ points on $C$. Given a set $S$ of $\F_q$ points disjoint from the support of $D$, we define $\AG(C,D,S)$ to be evaluation code $\Eval(\mathcal{L}(D),S)$. 
\end{definition}

We note that by Corollary~\ref{cor:genptsMDS} $\Eval(\mathcal{L}(D),S)$ up to a generic choice of $S$ will be $\MDS$. Hence by Corollary~\ref{cor:genptsMDS} over the algebraic closure for a generic choice of $S$ algebraic-geometric codes will be $\MDS(\ell)$ and $\LDMDS(\ell)$. But we want to work with constant sized fields over which generic arguments will not naively apply. We remedy this in the next subsection.

Every projective curve $C$ has an invariant associated to it called its genus $g$. Algebraically, it can be defined as one minus the constant term of the Hilbert polynomial of the curve (geometrically for complex projective curves it counts the number of `holes'). The main theorem about AG codes in the literature is the following.

\begin{theorem}[Theorem 21 from \cite{AGcodesurvey}]\label{thm:surveyAG}
For a projective curve $C$ of genus $g$, let $D$ be a divisor of $\F_q$ points of degree $d$ and $S$ a set of $\F_q$ points disjoint from the support of $D$ such that $d<|S|=n$. $\AG(C,D,S)$ will be a $[n,k=\dim_{\F_q} \mathcal{L}(D)]$ code with the following properties,
\begin{enumerate}
\item $k\ge d+1-g$,
\item if $2g-2<d$ then $k=d+1-g$,
\item the minimum distance of the code is at least $n-d$.
\end{enumerate}
\end{theorem}

To prove our theorems about punctured AG codes we need two more algebraic geometric preliminaries and a few lemmas.

First we need the fact that for a divisor $D$ of large enough degree the evaluation $\mathcal{L}(D)$ over a generic point gives us a closed embedding of the curve in projective space (for experts we want $\mathcal{L}(D)$ to be very ample).

\begin{lemma}[See 15.2.A, 18.4.1, 18.6.H, and 19.2.E in \cite{raviVakilFOAG}]\label{lem:evalEmbedProj}
Given a projective irreducible curve $C$ of genus $g$  and a divisor $D$ of degree $d$ on it, if $d\ge 2g+1$ then image of evaluating $\mathcal{L}(D)$ on points in $C$ defines an irreducible projective curve $C'$ of degree $d$ in $\Proj\overline{F}_q^{d-g}.$
\end{lemma}

The above lemma allows us to treat the columns of an AG code as points in an irreducible projective curve allowing us to apply Theorem~\ref{thm-mdsIrred}. We get a point in $d-g$ projective space because our space of functions is $d-g+1$ dimensional by Theorem~\ref{thm:surveyAG}.

Next we apply a union bound for a particular family of curves to prove our list decoding theorem. First we define and state the properties of the AG code we will use.

\paragraph{The Garcia-Stichtenoth tower}
We briefly summarize some facts, more details can be found in \cite{GScode}. Let $p$ be a prime power and $q=p^2$. We define a series of curves using the relations,
$$x_{i+1}^p+x_{i+1}=\frac{x_i^{p}}{x_i^{p-1}+1},i=1,\hdots,t-1.$$ 

We let the projectivized curve be $G_{p,t}$. There is only one $\F_q$ point at infinity. The affine curve has $N_q(G_{p,t})=p^{t-1}(p^2-p)$, $\F_q$ rational points. The genus $g(G_{p,t})$ is 
\begin{align*}
(p^{t/2}-1)^2 & \text{ if }t\text{ is even,}\\
(p^{(t+1)/2}-1)(p^{(t-1)/2}+1) & \text{ if }t\text{ is odd.}\\
\end{align*} 
We will use the bound $p^t+p^{(t+1)/2}-p^{(t-1)/2}-1\ge g(G_{p,t})\ge p^t -2p^{t/2}+1$.

Consider the divisor $sP_\infty$ with $s\ge 2g+1$. By Theorem~\ref{thm:surveyAG}, $\mathcal{L}(sP_\infty)$ is $s-g(G_{p,t})+1$ dimensional. We use the $N_q(G_{p,t})$ $\F_q$ rational points as evaluation points. Therefore, $G_{p,e}$ with the divisor $sP_\infty$ defines a $[N_q(G_{p,t}),s-g(G_{p,t})+1]$ code with distance $N_q(G_{p,t})-s$. We call this code $GS(p,t,s)$. We also note by Lemma~\ref{lem:evalEmbedProj} the columns of the generating matrix of this code are coming from an irreducible projective curve of degree $s$. By Corollary~\ref{cor:genptsMDS} this curve is $\MDS$ because the functions being evaluated are linearly independent.

\begin{theorem}\label{thm:AGpuncture}
Let $\epsilon>0$, $R\in (0,1)$ be rational numbers. If
\begin{align*}
p\ge 1+4/R+12eL\epsilon^{-1} 2^{5(L+2)/\epsilon}, && \text{\em  (If we do {\bf not} allow repetitions in } G\text{\em )}\\
p\ge 1+12eL\epsilon^{-1} 2^{5(L+2)/\epsilon},&&\text{\em  (If we allow repetitions in } G\text{\em )}
\end{align*}

then for any large enough $n>0$, such that there exists an integer $t$ satisfying $Rn/2\ge p^t\ge Rn/4$, a $[n,nR]$ code whose generator matrix is constructed by sampling the columns of the generator of $GS(p,t,s)$ with $s=Rn+g(G_{p,t})-1$ will give a $\rLDMDS_{\epsilon n}(\le L)-$ code over a field of size $p^2$ with probability at least at least $1-1/2^{Ln}$.   
\end{theorem} 

We note the theorem is only meaningful if an infinite sequence of such $n$ can be found. That is true and easy to check. We also note that consecutive such values of $n$ are at most a $p$ multiplicative factor apart from each other.
 
\begin{proof}
We only prove the statement for the no repetitions allowed case. The calculation is nearly identical when we allow repetitions.

Set $k=s-g(G_{p,e})+1$ with $s= Rn+g(G_{p,t})-1$. We need $s\ge 2g(G_{p,e})+1$ so that code can be treated as coming from an degree $s$ $\MDS$ curve, which means we need $Rn\ge g(G_{p,e})+2$ which is satisfied by our assumption that $Rn\ge 2p^t$ and the fact that $g(G_{p,e})\le p^t+p^{(t+1)/2}-p^{(t-1)/2}-1$.
To apply Theorem~\ref{thm:probBoundSlackLower} we want to bound $P_q(X,L)$. As $X$ is an irreducible curve of degree $s$ (by Lemma~\ref{lem:evalEmbedProj}) we have by Bezout's Theorem (Lemma~\ref{lem:bezout}), $P_q(X,L)\le sL$.
Applying Theorem~\ref{thm:probBoundSlackLower} we see that for the punctured code to be $\rLDMDS_{\epsilon n}(\le L)$ the failure probability is at most
\begin{equation}\label{eq:AGLDSlack1}
2^{(L+1)n}\binom{n}{\epsilon n /2}2^{\epsilon n (L+1)/2} \left(\frac{sL}{N_q(G_{p,t})-n}\right)^{\epsilon n /2}.
\end{equation}
We note,
\begin{align*}
\frac{sL}{N_q(G_{p,t})-n}&\le \frac{L(nR+p^t+p^{(t+1)/2})}{p^{t-1}(p^2-p)-n}\\
&\le  \frac{L(1+p^{-(t-1)/2}+nR/p^{t})}{p-1-n/p^t} \\
&\le \frac{6L}{p-1-4/R}. && (\text{Using }p^t\ge Rn/4)
\end{align*}

Using this in \eqref{eq:AGLDSlack1} we can upper bound the failure probability by

$$2^{(L+2)n}\left(\frac{en}{\epsilon n/2}\right)^{\epsilon n / 2} 2^{\epsilon n (L+1)/2} \left( \frac{6L}{p-1-4/R}\right)^{\epsilon n/2}\le 2^{-Ln},$$
for $p\ge 1+4/R+12eL/\epsilon 2^{5(L+2)/\epsilon}.$
\end{proof}

By choosing the right list size we now show that punctured AG codes can achieve list decoding capacity.

\begin{corollary}\label{cor:AGListcap}
Let $\epsilon,R\in (0,1)$ be rationals. Let $L=\lceil 2\frac{1-R}{\epsilon}\rceil$ and $p$ be a prime such that
$$\log_2p\gg_{R} 1/\epsilon^2, $$
any large enough $n>0$, such that there exists an integer $t$ satisfying $Rn/2\ge p^t\ge Rn/4$, a random puncturing of $GS(p,t,s)$ with $s=Rn+g(G_{p,t})-1$ will give a $(1-R-\epsilon,L)$ average-radius list-decodable $[n,nR]$-code over a field of size $p^2$ with probability at least $1-1/2^{Ln}$.   
\end{corollary}
\begin{proof}
Apply the previous theorem for $\epsilon/2$ and list size $L=\lceil 2\frac{1-R}{\epsilon}\rceil$.
\end{proof}

A similar calculation now applies to the upper $\MDS$ relaxation (or equivalently the relaxed $\MR$  tensor code setting).
\begin{theorem}\label{thm:AGpunctureMRRate}
Let $\epsilon>0$, $R\in (0,1)$ be rational numbers. If 
$$p\ge 1+4/R+12e(L-1)\epsilon^{-1} 2^{5(L+1)/\epsilon},$$
then for any large enough $n>0$, such that there exists an integer $t$ satisfying $Rn/2\ge p^t\ge Rn/4$ and $s=Rn+g(G_{p,t})-1$ a random puncturing of $GS(p,t,s)$ will give a $\rMDS^{\epsilon n}(L)-[n,nR]$-code over a field of size $p^2$ with probability at least $1-1/2^{Ln}$.   
\end{theorem} 

We note that we can not prove a similar result in the constant co-dimension setting. This is because we need $Rn/p^t$ to be bounded away from $0$ which is not possible if $p$ is already polynomial in $n$.

\section*{Acknowledgments}
We would like to thank Omar Alrabiah, Zeyu Guo, Venkatesan Guruswami and Ray Li for discussions that inspired this work. We also thank Zeyu Guo for clarifying that the constructions of \cite{guo2021efficient} and \cite{guruswami2022optimal} are non-linear. We thank anonymous referees for helpful comments.

The first author was supported in part by a Microsoft Research PhD fellowship while a student at Stanford University. The second author was supported by NSF grant DMS-1953807 while this author was a graduate student at the Department of Computer Science, Princeton University.

\appendix

\section{Omitted Proofs}\label{app:omit}

\subsection{Lemma~\ref{lem:dual-inter-rank}}

\lemdual*

\begin{proof}
  By the rank-nullity theorem, we have that
  \[
    \dim\coker  \cH_{A_1,\hdots,A_\ell}[G] = n + \ell k - \rank M,
  \]
  where $\coker$ denotes the left null space of a matrix; i.e.,
  \[
    \coker M := \{x \in \F^m : x^T M = 0\}.
  \]

  Thus, it suffices to prove that
  \begin{align}
    \dim(H_{A_1} \cap \cdots \cap H_{A_\ell}) &= \dim\coker  \cH_{A_1,\hdots,A_\ell}[G] - \sum_{i=1}^{\ell} (k - \rank(G|_{\oA_i})) - k\nonumber\\
    &= \dim\coker  \cH_{A_1,\hdots,A_\ell}[G] - \sum_{i=1}^{\ell} \dim\coker(G|_{\oA_i}) - k.\label{eq:684}
  \end{align}
  Consider $(y, x_1, \hdots, x_\ell) \in \F^{n+\ell k}$ where $y \in \F^n$ and each $x_i \in \F^k$. Note that $(y, x_1, \hdots, x_\ell) \in \coker  \cH_{A_1,\hdots,A_\ell}[G]$ if and only if for all $i \in [\ell]$, we have that
\[
  (G^Tx_i + y)|_{\oA_i} = 0.
\] 

Let $Y \subset \coker  \cH_{A_1,\hdots,A_\ell}[G]$ denote the vector space of tuples for which $x_1 = \cdots = x_\ell$ and $y = -G^{T}x_1.$ Note that $\dim Y = k$.

For all $i \in [\ell]$, let $X_i \subset \coker  \cH_{A_1,\hdots,A_\ell}[G]$ denote the vector space of tuples for which $y = 0$ and $x_j = 0$ for all $j \in [\ell] \setminus \{i\}$. Then, the defining condition is that $(G^Tx_i)|_{\oA_i} = 0$. In other words, $x_i \in \coker(G^T|_{\oA_i})$. Therefore, to prove (\ref{eq:684}), it suffices to establish an isomorphism
\[
  H_{A_1} \cap \cdots \cap H_{A_\ell} \cong (\coker  \cH_{A_1,\hdots,A_\ell}[G]) / (Y + X_1 + \cdots + X_\ell)
\]
Consider the map $\psi : \coker  \cH_{A_1,\hdots,A_\ell}[G] \to H_{A_1} \cap \cdots \cap H_{A_\ell}$ defined by \[\psi(y, x_1, \hdots, x_\ell) := Hy.\] To finish, it suffices to prove the following three properties about $\psi$.
\begin{enumerate}
  \item[(a)] $\psi$ is well-defined
  \item[(b)] $\ker \psi = Y + X_1 + \cdots + X_\ell$.
  \item[(c)] $\psi$ is surjective.
\end{enumerate}

To prove (a), consider $(y, x_1, \hdots, x_\ell) \in \coker  \cH_{A_1,\hdots,A_\ell}[G]$. Let $c_i := G^{T}x_i$ and $z = Hy$. Note that $Hc_i = 0$, since $H$ is a parity check matrix. Thus, since $(c_i + y)|_{\oA_i} = 0$,
\[
  z = Hy = H(c_i + y) = H|_{A_i}(c_i + y)|_{A_i}.
\]
That is, $z \in H_{A_i}$, so $z \in H_{A_1} \cap \cdots \cap H_{A_\ell}$. Thus, the image of $\psi$ is contained in $H_{A_1} \cap \cdots \cap H_{A_\ell}$.

To prove (b), consider $(y, x_1, \hdots, x_\ell) \in \coker  \cH_{A_1,\hdots,A_\ell}[G]$ with $Hy = 0$. Thus, $y$ must be a codeword, so there exists a unique $x \in \F^k$ with $y = -G^Tx$. For all $i \in [\ell]$ we have that $(G^T(x_i - x))|_{\oA_i} = 0$. This holds if and only if $x_i \in x + \coker G^T|_{\oA}$. This establishes that $\ker \psi = Y + X_1 + \cdots + X_\ell$.

To prove (c), consider any $z \in H_{A_1} \cap \cdots \cap H_{A_\ell}$. For each $i \in [\ell]$, since $z \in H_{A_i}$, there is at least one $y_i \in \F^n$ with $(y_i)|_{\oA_i} = 0$ and $Hy_i = z$. Let $c_i = y_i - y_1$ (so $c_1 = 0$). Since $Hc_i = 0$, we have that there exists $x_i \in \F^k$ with $c_i = G^{T}x_i$. We claim that $( y_1, x_1, \hdots, x_\ell) \in \coker  \cH_{A_1,\hdots,A_\ell}[G]$. To see why, for all $i \in [\ell]$, we have that
\[
  (G^{T}x_i + y_1)|_{\oA_i} = (c_i + y_1)|_{\oA_i} = (y_i)|_{\oA_i} = 0.
\]
This completes that $\psi$ is an isomorphism, so the dimension formula holds.
\end{proof}

\subsection{Proposition~\ref{prop:mr-tensor}}

\propmrtensor*

\begin{proof}
We prove (a) iff (b) and then (b) iff (c).

\paragraph{(a) iff (b).} Note that $E$ is correctable if and only if the map $x \mapsto x|_{\bar{E}}$ is injective for $x \in C_1 \otimes C_2$. This map is injective if and only if $(U \otimes V)|_{\bar{E}}$ has rank $\dim (C_1 \otimes C_2) = (m-a)(n-b)$. Since $U \otimes V$ has $(m-a)(n-b)$ rows, this is true if and only if the columns of $(U \otimes V)|_{\bar{E}}$ span $\F^{m-a} \otimes \F^{n-b}$.

\paragraph{(b) iff (c).} We now assume $a=1$ and $C_1$ is MDS. Thus, up to scaling, there is a unique vector $\lambda \in \F^m \setminus \{0\}$ such that $\sum_{i=1}^{m} \lambda_i U_i = 0$. Further, since $C_1$ is MDS, each $\lambda_i \neq 0$. Note that (b) is equivalent to.
\begin{align}
\rank \begin{pmatrix}
I_{n-b} & \\
& (U \otimes V)|_{\bar{E}}
\end{pmatrix} = m(n-b).\label{eq:rank-34}
\end{align}
Thus, it suffices to identify row/column operations which transforms the matrix (\ref{eq:rank-34}) into $\cG_{A_1, \hdots, A_m}[V]$.

First, consider an invertible linear map $\psi : \F^{m-1} \to \F^{m-1}$ for which $\psi(U_i) = e_i$ for $i \in [m-1]$. Since, $\sum_{i=1}^{m} \lambda_i U_i = 0$, we have that $\psi(U_{m}) = \sum_{i=1}^{m-1} -\lambda_i e_i / \lambda_m$. In particular, by suitable row operations, we can transform (\ref{eq:rank-34}) into
\begin{align}
\rank \begin{pmatrix}
I_{n-b} & & & &\\
& V_{A_1} & & & -(\lambda_1 / \lambda_m) V_{A_m}\\
& & \ddots & & \vdots \\
& & & V_{A_{m-1}} & -(\lambda_{m-1} / \lambda_m) V_{A_m}.
\end{pmatrix} = m(n-b).\label{eq:rank-56}
\end{align}
Since each $\lambda_i \neq 0$, we can suitably scale the row/columns blocks of (\ref{eq:rank-56}) to get.
\begin{align*}
\rank \begin{pmatrix}
I_{n-b} & & & &\\
& V_{A_1} & & & -V_{A_m}\\
& & \ddots & & \vdots \\
& & & V_{A_{m-1}} & -V_{A_m}.
\end{pmatrix} = m(n-b).
\end{align*}
Next, we add a copy of the first row block to each of the other row blocks to get
\begin{align*}
\rank \begin{pmatrix}
I_{n-b} & & & &\\
I_{n-b} & V_{A_1} & & & -V_{A_m}\\
\vdots & & \ddots & & \vdots \\
I_{n-b} & & & V_{A_{m-1}} & -V_{A_m}.
\end{pmatrix} = m(n-b).
\end{align*}
Finally, we add suitable linear combinations of the columns in the first column block to the columns in the last column block to get $\cG_{A_1, \hdots, A_m}[V]$ (up to permutation of the row blocks), as desired.
\end{proof}

We show an upper bound on the number of patterns which needed to be verified to be correctable in order to verify that all patterns within $\mathcal E^{m,n}_{a,b}$ are correctable. The particular bound we prove here is essentially the upper bound proved in \cite{bgm2021mds}. We adopt the notation $\binom{a}{\le b}=\sum_{i=0}^b \binom{a}{i}.$

\subsection{Lemma~\ref{lem:noOfMRpatterns}}

\lemnoOfMRpatterns*

\begin{proof}
The $2^{mn}$ upper bound is trivial as $|\mathcal E^{m,n}_{a,b}| \le 2^{mn}$.

Given a pattern $E \in \mathcal E^{m,n}_{a,b}$, let $A_1, \hdots, A_m \subseteq [n]$ be the rows of $E$ and $B_1, \hdots, B_n \subseteq [m]$ be the columns of $E$. Note that if some $|A_i| \le a$, then using the decoding properties of $C_1$, we can replace WLOG that $A_i = \emptyset$. Likewise, if $|B_j| \le b$, then we can replace WLOG with $B_j = \emptyset$. We call such a pattern \emph{reduced} (c.f., \cite{Gopalan2016}).

Let $\mathcal P^{m,n}_{a,b}$ be the set of all reduced patterns. For a given $E \in \mathcal P^{m,n}_{a,b}$. Let $m'$ be the number of nonempty rows and $n'$ the number of nonempty columns. If $E = \emptyset$, there is nothing to check. Otherwise, $E$ has at least $a$ rows and $b$ columns. Therefore, by Theorem~\ref{thm:regularity}, we have that $|E| \le bm'+an'-ab$. Further, we can bound that, $|E| \ge \max(m'(b+1),n'(a+1))$. Thus, we get that $n'(a+1) \le bm'+an'-ab$ or $n' \le b(m'-a) \le b(m-a)$. Therefore, we can count all of $\mathcal P^{m,n}_{a,b}$ in the following way:

\begin{itemize}
\item Choose at most $b(m-a)$ columns $C \subseteq [n]$.
\item Pick $E \subseteq C \times [n]$ of size at most $bm + a(b(m-a)) -ab = b(a+1)(m-a)$. 
\end{itemize}

This proves the second upper bound. The third upper bound follows by swapping $m$ and $n$.
\end{proof}

\subsection{Proposition~\ref{prop:lower-mds-alt}}

\propA*

\begin{proof}
  For the ``only if'' direction, we need to prove (1) and (2). To prove (1), for each $i \in [n]$ consider $A_1 = \{i\}$ and $A_2 = \cdots = A_\ell = \{\}$. Then, one can easily verify that $\mathcal G_{A_1, \hdots, A_\ell}[W] = k+1$, so $\mathcal G_{A_1, \hdots, A_\ell}[G] = k+1$ by assumption. This latter condition is only possible if $G_i$ is nonzero. To prove (2), consider  $A_1, \hdots, A_\ell \subseteq [n]$ of size at most $k-d$ with the $k-d$-dimensional null intersection property. By Proposition~\ref{prop:tian-rank}, we have that $\cG_{A_1, \hdots, A_\ell}[W]$ has full column rank. Thus, $\cG_{A_1, \hdots, A_\ell}[G]$ has full column rank, so by Proposition~\ref{prop:tian-rank} again, we have that $G_{A_1} \cap \cdots \cap G_{A_\ell} = 0$.

  For the ``if'' direction, consider $A_1, \hdots, A_\ell$ such that $\mathcal G_{A_1, \hdots, A_\ell}[W]$ has full column rank. By Proposition~\ref{prop:tian-rank}, we must have each $|A_i| \le k-d$ and $A_1, \hdots, A_\ell$ have the $k-d$-dimensional null intersection property. Thus, $G_{A_1} \cap \cdots \cap G_{A_\ell} = 0$. Therefore, by Proposition~\ref{prop:tian-rank}, it suffices to show that $\rank G_{A_i} = |A_i|$ for all $i$.

  We shall prove by induction on $|A|$ that $\dim(G_{A}) = |A|$. If $|A| \le 1$, then the columns $G|_{A}$ are linearly independent by assumption that every column is nonzero. Otherwise, $|A| \ge 2$, so we partition $A = A_1 \cup A_2$ with each $|A_1|, |A_2| \le |A|-1$. Then, observe that
\[
  \dim(G_A) = \dim(G_{A_1} + G_{A_2}) = \dim(G_{A_1}) + \dim(G_{A_2}) - \dim(G_{A_1} \cap G_{A_2}) = |A_1| + |A_2| - 0 = |A|,
\]
where we use the induction hypothesis and that $A_1, A_2$ have the $k-d$-dimensional null intersection property.
\end{proof}

\subsection{Proposition~\ref{prop:lower-mds}}

\propB*

\begin{proof}
\begin{itemize}
\item[(a)] Follows directly from Proposition~\ref{prop:2.2} and Definition~\ref{def:lower-mds}.
\item[(b)] Follows from the fact that if $A_1, \hdots, A_\ell \subseteq [n]$ have the $k-d'$-dimensional null intersection property, then $A_1, \hdots, A_\ell$ have the $k-d$-dimensional null intersection property.
\item[(c)] Consider any $A_1, \hdots, A_{\ell'} \subseteq [n]$ with the $k-d$-dimensional null intersection property. Let $A'_1, \hdots, A'_\ell \subseteq [n]$ be a family of sets with $A'_i = A_i$ for all $i \in [\ell'-1]$ and $A'_{\ell'} = \cdots = A'_\ell = A_{\ell'}$. Clearly $A'_1, \hdots, A'_\ell$ have the $k-d$-dimensional null intersection property. Since $C$ is $\rMDS_d(\ell)$, we have that
\[
   0 = G_{A'_1} \cap \cdots \cap G_{A'_\ell} = G_{A_1} \cap \cdots \cap G_{A_{\ell'}},
\]
as desired.
\item[(d)] Consider $A_1, A_2 \subseteq [n]$ of size at most $k-d$ with the $k-d$-dimensional null intersection property. By Proposition~\ref{prop:null}, we have that $|A_1| + |A_2| \le k-d$ and $|A_1 \cap A_2| = 0$. Thus, $|A_1 \cup A_2| \le k-d$.

If every $k-d$ columns of $G$ are linearly independent, then
\begin{align*}
  \dim(G_{A_1} \cap G_{A_2}) &= \dim(G_{A_1}) + \dim(G_{A_2}) - \dim(G_{A_1} + G_{A_2})\\
                             &= |A_1| + |A_2| - |A_1 \cup A_2|\\
                             &= 0.
\end{align*}
Therefore, $C$ is $\rMDS_d(2)$. Conversely, assume that $C$ is $\rMDS_d(2)$. Consider $A \subseteq [n]$ of size at most $k-d$. We shall prove by induction on $|A|$ that $\dim(G_{A}) = |A|$. If $|A| \le 1$, then the columns $G|_{A}$ are linearly independent by assumption that every column is nonzero. Otherwise, $|A| \ge 2$, so we partition $A = A_1 \cup A_2$ with each $|A_1|, |A_2| \le |A|-1$. Then, observe that
\[
  \dim(G_A) = \dim(G_{A_1} + G_{A_2}) = \dim(G_{A_1}) + \dim(G_{A_2}) - \dim(G_{A_1} \cap G_{A_2}) = |A_1| + |A_2| - 0 = |A|,
\]
where we use the induction hypothesis and that $A_1, A_2$ have the $k-d$-dimensional null intersection property.
\end{itemize}
\end{proof}

\subsection{Proposition~\ref{prop:upper-mds}}

\propC*

\begin{proof}
\begin{itemize}
\item[(a)] Follows directly from Proposition~\ref{prop:MDS-sat} and Definition~\ref{def:upper-mds}.
\item[(b)] Consider any $A_1, \hdots, A_\ell \subseteq [n]$ have the $k+d'$-dimensional saturation property. Let $W$ be a generic $(k+d') \times n$ matrix. Then, we have that $\rank \cG_{A_1, \hdots, A_\ell}[W|_{[k+d] \times [n]}] = \ell (k + d')$. If we delete the last $d'-d$ rows from each of the $\ell$ row blocks of $M$, we are left with \[
    \rank \cG_{A_1, \hdots, A_\ell}[W|_{[k+d] \times [n]}] \ge \ell (k+d') - \ell(d'-d) = \ell (k+d).
  \]
  Thus, $A_1, \hdots, A_\ell$ have the $k+d$-dimensional saturation property. (Note the rank of $M$ cannot exceed saturation.) Thus, $C$ is $\rMDS^{d'}(\ell)$.
\item[(c)]  Consider any $A_1, \hdots, A_{\ell'} \subseteq [n]$ with the $k+d$-dimensional saturation property. Set $A_{\ell'+1} = \cdots = A_{\ell} = [n]$. One can verify by Proposition~\ref{prop:tian-rank} that $A_1, \hdots, A_\ell$ have the $k+d$-dimensional saturation property. Thus, $\rank \cG_{A_1, \hdots, A_\ell}[G] = \ell (k+d)$. By deleting the last $(\ell - \ell')(k+d)$ rows, we get that $\rank \cG_{A_1, \hdots, A_{\ell'}}[G] \ge \ell' (k+d)$. Thus, $A_1, \hdots, A_{\ell'}$ is $G$-saturated so $C$ is $\rMDS^d(\ell')$.

\item[(d)] First, assume for some $A \subseteq [n]$ of size $k+d$, $G_A \neq F^k$. It is clear that $A_1 = A_2 = A$ has the $k+d$-dimensional saturation property, but by Proposition~\ref{prop:tian-rank},
\[
  \rank \begin{pmatrix}
I_k & V|_{A} &\\
I_k & & V|_{A}
\end{pmatrix} = k + \dim V|_A < 2k,
\]
a contradiction. Thus, if $C$ is $\rMDS^d(2)$, every $k+d$ columns of $G$ span $\F^k$.

For the other direction, assume every $k+d$ columns of $G$ span $\F^k$. Consider any $A_1, A_2 \subseteq [n]$ with the $k+d$-dimensional saturation property. Let $W$ be a generic matrix, by Proposition~\ref{prop:tian-rank}, we then have that $\rank \cG_{A_1, A_2}[W] = 2k$ is equivalent to
\[
  k = \dim(W_{A_1}) + \dim(W_{A_2}) - \dim(W_{A_1} \cap W_{A_2}) = \dim(W_{A_1} + W_{A_2}) = \min(|A_1 \cup A_2|, k).
\]
Thus, $|A_1 \cup A_2| \ge k$, so $\dim(V_{A_1 \cup A_2}) = k$ by assumption. Therefore,
\[
  \rank \cG_{A_1, A_2}[V] = k + \dim(V_{A_1}) + \dim(V_{A_2}) - \dim(V_{A_1} \cap V_{A_2}) = k + \dim(V_{A_1 \cup A_2}) = 2k.
\]
Hence, $C$ is $\MDS^d(2)$, as desired.
\end{itemize}
\end{proof}

\bibliographystyle{alpha}
\bibliography{references}

\end{document}